\documentclass[11pt]{article}

\usepackage{fullpage}
\usepackage{graphicx}
\usepackage{amsfonts}
\usepackage{amsthm,amsmath, mathtools, enumitem, amssymb}
\usepackage{bbm}
\usepackage{caption}
\usepackage{subcaption}
\usepackage[multiple]{footmisc}
\usepackage{color}
\usepackage{physics}
\usepackage{comment}
\usepackage{multirow}
\usepackage{appendix}

\usepackage[table, dvipsnames]{xcolor}
\definecolor{blue_color}{HTML}{088CE8}

\newcommand{\MAJ}{{\mathsf{MAJ}}}

\newcommand{\eat}[1]{}

\newcommand{\adeg}{\widetilde{\deg}}

\newcommand{\eps}{\varepsilon}
\newcommand{\zbits}{\{0,1\}}

\newcommand{\bits}{\{-1,1\}}

\newcommand{\reals}{\mathbb{R}}

\newcommand{\OR}{\mathsf{OR}}

\newcommand{\E}{\mathbb{E}}
\newcommand{\poly}{\operatorname{poly}}
\newcommand{\polylog}{\operatorname{polylog}}

\newcommand{\OS}{\mathsf{OS}}
\newcommand{\IP}{\mathsf{IP}}
\newcommand{\GT}{\mathsf{GT}}
\newcommand{\HS}{{\mathsf{HS}}}
\newcommand{\HSA}{{\AHS}}
\newcommand{\AHS}{{\mathsf{AHS}}}
\newcommand{\parity}{{\mathsf{parity}}}
\newcommand{\R}{\mathcal{R}}
\newcommand{\CRAP}{{\mathsf{PUR}}} 
\newcommand{\CRAPWithParameter}[1]{{\mathsf{PUR[#1]}}}
\newcommand{\Oracle}{\mathcal{Y}}
\newcommand{\OracleWithParameter}[1]{\mathcal{Y}{[#1]}}
\newcommand{\oracle}{y}
\newcommand{\from}{\leftarrow}
\newcommand{\one}{\mathbf{1}}

\theoremstyle{definition}
\usepackage{hyperref}
\usepackage{cleveref}

\newtheorem{theorem}{Theorem}
\newtheorem*{theorem*}{Theorem}
\newtheorem*{corollary*}{Corollary}
\newtheorem*{statement*}{Statement}
\newtheorem{definition}[theorem]{Definition}

\newtheorem{claim}[theorem]{Claim} 
\newtheorem{corollary}[theorem]{Corollary}

\title{Approximate degree lower bounds for oracle identification problems}
\author{Mark Bun\footnote{Department of Computer Science, Boston University, 
Boston, MA 02215, USA\@. {\tt mbun@bu.edu}. Supported in part by NSF awards CCF-1947889 and CNS-2046425 and a Sloan Research Fellowship.} \and Nadezhda Voronova\footnote{Department of Computer Science, Boston University, 
Boston, MA 02215, USA\@. {\tt voronova@bu.edu}. Supported in part by NSF awards CCF-1947889 and CNS-2046425. 
}}
\date{March 2023}

\newcommand{\NV}[1]{{\color{blue} (Nadya: #1)}}

\begin{document}

\maketitle

\begin{abstract}
    The approximate degree of a Boolean function is the minimum degree of real polynomial that approximates it pointwise. For any Boolean function, its approximate degree serves as a lower bound on its quantum query complexity, and generically lifts to a quantum communication lower bound for a related function.

    We introduce a framework for proving approximate degree lower bounds for certain oracle identification problems, where the goal is to recover a hidden binary string $x \in \{0, 1\}^n$ given possibly non-standard oracle access to it. Our lower bounds apply to decision versions of these problems, where the goal is to compute the parity of $x$.
    We apply our framework to the ordered search and hidden string problems, proving nearly tight approximate degree lower bounds of $\Omega(n/\log^2 n)$ for each. These lower bounds generalize to the weakly unbounded error setting, giving a new quantum query lower bound for the hidden string problem in this regime.  
    Our lower bounds are driven by randomized communication \emph{upper bounds} for the greater-than and equality functions.
\end{abstract}

\section{Introduction}

In an \emph{oracle identification} problem, there is an unknown string $x \in \zbits^n$. A query algorithm is given possibly non-standard oracle access to $x$, and its goal is to reconstruct $x$ by making a minimal number of queries to this oracle. More specifically, an oracle identification problem is specified by a fixed family of Boolean functions $a_1, \dots, a_N$. A query algorithm may inspect any value $a_i(x)$ of its choice at the cost of one query, and its goal is to determine $x$. Many influential problems in the study of quantum algorithms and complexity can be viewed as oracle identification problems, including van Dam's original oracle interrogation problem~\cite{Dam98}, the Bernstein-Vazirani problem~\cite{BernsteinV93}, combinatorial group testing~\cite{AmbainisM14, Belovs15}, symmetric junta learning~\cite{Belovs15}, and more~\cite{BuhrmanW99, AmbainisIKMPY04, AmbainisIKRY07, IwamaNRT12, cleve2012reconstructing, Kothari14}. In this work, we study two such oracle identification problems:

\paragraph{Ordered Search.}
Consider the following abstraction of the problem of searching an ordered list of $N = 2^n$ elements. Given a list of $N$ bits $a_i \in \{0, 1\}$ under the promise that $a_0 \le a_1 \le \dots \le a_{N-1}$, find the (binary encoding of the) minimum index $x \in \zbits^n$ such that $a_x = 1$. Binary search yields a deterministic algorithm making $n$ queries, and it is not hard to see that this is optimal for randomized algorithms as well. As for quantum algorithms, it turns out that a constant-factor speedup is possible~\cite{FarhiGGS99, ChildsLP07, BenOrH08}, but a lower bound of $\Omega(n)$ holds in this model as well~\cite{BuhrmanW99, FarhiGGS98, Ambainis99, HoyerNS02, ChildsL08}. Ordered search may be viewed as an oracle identification problem where the query algorithm is given oracle access to $a_0 = \GT_0(x), \dots, a_{N-1} = \GT_{N-1}(x)$, where each ``greater-than'' function $\GT_i(x)$ evaluates to $1$ if $i \ge x$ and to $0$ otherwise. 

\paragraph{Hidden String.}
In the hidden string problem, the goal is to reconstruct a hidden string $x \in \zbits^n$ given information about the presence of absence of potential substrings of $x$. That is, the goal is to determine $x$ given ``substring oracle'' access, i.e., oracle access to $a_s = \phi_s(x)$ for every binary string $s$ of length at most $n$, where $\phi_s(x)$ evaluates to $1$ iff $s$ is a substring of $x$. Building on a classical query algorithm of Skiena and Sundaram~\cite{SkienaS95}, Cleve et al.~\cite{cleve2012reconstructing} gave a $3n/4 + o(n)$ quantum query algorithm for this problem, and proved a nearly matching quantum query lower bound of $\Omega(n / \log^2 n)$.

\medskip

The state-of-the-art quantum query lower bounds for both problems are proved via the quantum adversary method, which in its modern formulation~\cite{HoyerLS07}, characterizes the bounded-error quantum query complexity of every function up to a constant factor~\cite{Reichardt11}. The other major technique for proving quantum query lower bounds is the polynomial method~\cite{beals2001quantum}, which lower bounds the quantum query complexity of a function by lower bounding its \emph{approximate degree}. The approximate degree of a Boolean function is the least degree of a real polynomial that approximates it pointwise to error $1/3$. Since the acceptance probability of a $T$-query quantum algorithm is a polynomial of degree $2T$, the approximate degree of a function is always at most (half of its) quantum query complexity, but it can be much smaller~\cite{Ambainis06, AaronsonBK16, Sherstov20, BunKT20}.

In this work, we prove lower bounds of $\Omega(n / \log^2 n)$ on the approximate degree of (decision variants) of the ordered search and hidden string problems. These lower bounds are nearly optimal, as the known quantum (indeed, even classical) query algorithms for these problems automatically yield $O(n)$ upper bounds on their approximate degree. For the ordered search problem, Childs and Lee~\cite{ChildsL08} explicitly posed the question of investigating approximate degree lower bounds to circumvent limitations of the adversary method.  Meanwhile, our lower bound on the approximate degree of the hidden string problem implies a quantum query lower bound matching the state-of-the-art~\cite{cleve2012reconstructing}.

Approximate degree is a fundamental measure of the complexity of Boolean functions that has been the subject of extensive study in its own right (see, e.g.,~\cite{BunT22} for a recent survey). And while nearly tight quantum query lower bounds for these problems were already known, we see two main quantum motivations for recovering these bounds via approximate degree. First, there are senses in which approximate degree is a more robust lower bound technique than the adversary method. For example, via Sherstov's pattern matrix method~\cite{Sherstov11Pattern}, any approximate degree lower bound for a Boolean function $f$ can be ``lifted'' to give the same quantum communication lower bound for a related two-party function $F$. Such a generic lifting result is not known for any other general quantum query lower bound technique. Moreover, variants of the polynomial method are capable of proving lower bounds against zero-, small-, and unbounded-error quantum algorithms~\cite{beals2001quantum, BuhrmanCWZ99}, as well as time-space tradeoffs~\cite{KlauckSW07}. Indeed, using the polynomial method, we give weakly-unbounded-error quantum query lower bounds for the hidden string problem (see Corollary~\ref{cor:hs-ue-lb}) that significantly improve over the lower bound implied by the adversary method~\cite{cleve2012reconstructing}.

Second, we believe that our approximate degree lower bounds shed additional light on what makes the ordered search and hidden string problems hard, and may be more transparent in this regard than existing adversary lower bounds.
In particular, our lower bounds show that it is not only hard for quantum algorithms to reconstruct the hidden string $x$, but even to simply compute its parity (a decision problem). The other nearly tight lower bounds for the problems we consider appear to make essential use of the fact that the query algorithm needs to reconstruct all of $x$, and it isn't clear (at least to us) how to adapt them to hold for their decision variants. We believe that the technique we introduce, or at the very least the ``indirect'' method we use to prove our lower bounds, will be more broadly useful in understanding the approximate degree and quantum query complexity of other oracle identification problems.

\subsection{Techniques}

Here we give a brief summary of the ideas behind our lower bound for ordered search. A  more detailed technical overview, including a discussion of how we apply our framework to the hidden string problem, appears in Section~\ref{sec:roadmap}. Full proofs appear in Sections~\ref{sec:OS} and~\ref{sec:HS}.

The first lower bound for quantum ordered search was given by Buhrman and de Wolf~\cite{BuhrmanW99}, who actually showed an $\Omega(\sqrt{n})$ lower bound on its approximate degree.
The starting point for the proof of our lower bound is their ingenious indirect argument, so let us review it here. 
Recall that the ability to solve ordered search on inputs $a_0 \le a_1 \le \dots \le a_{N-1}$ enables recovering the string $x \in \{0, 1\}^n$, where $N = 2^n$, for which every $a_i = \GT_i(x)$. This, in particular, enables the evaluation of any ``hard'' Boolean function of $x$, e.g., its parity. In light of this, define the partial Boolean function  $\OS_N(a_{0}, \dots, a_{N-1}) := \parity(x)$ whenever there exists an $x$ for which $a_i = \GT_i(x)$ for every $i$. Let $p : \zbits^N \to \mathbb{R}$ be a polynomial of degree $d$ approximating $\OS_N$. It is known that every polynomial approximating $\parity$ must have degree $\Omega(n)$, so the goal now is to use this fact to prove a lower bound on the degree of $p$. To do so, we use the additional fact that the functions $\GT_i$ can each be approximated by a degree $O(\sqrt{n})$ polynomial $q_i$ arising from, say, a variant of Grover search. By making $p$ ``robust to noise'' in its input without increasing its degree~\cite{BuhrmanNRW07, sherstov2012making}, we get that the composed polynomial $p(q_{0}(x), \dots, q_{N-1}(x)) \approx \parity(x)$ and has degree $O(d\sqrt{n})$. Now the fact that the approximate degree of $\parity$ is $\Omega(n)$ implies that $d = \Omega(\sqrt{n})$.

In summary, the lower bound for $\OS_N$ follows from the fact that we can express the function $\parity(x) = \OS_N(\GT_{0}(x), \dots, \GT_{N-1}(x))$, where we have a lower bound on the approximate degree of $\parity$ and an \emph{upper bound} on the approximate degree of $\GT$. However, the lower bound gets stuck at degree $\Omega(\sqrt{n})$ because the functions $\GT_i$ themselves require nontrivial degree $O(\sqrt{n})$ to approximate, and this is tight.

To get an improved lower bound of $\widetilde{\Omega}(n)$ on the approximate degree of $\OS_N$, we introduce the following idea to make $\GT$ behave as if it were easier to approximate by low degree polynomials, while preserving the hardness of parity. Given an input $x \in \zbits^n$, we redundantly encode $x$ as a longer string $\Oracle(x) \in \zbits^m$ for some $m = \poly(n)$. This encoding is chosen so that 
\begin{itemize}
	\item Access to $\Oracle(x)$ instead of just $x$ itself makes each function $\GT_i(x)$ approximable by a much lower degree polynomial. That is, for every $i$, there exists a polynomial $q_i$ of degree $\polylog(n)$ such that $q_i(\Oracle(x)) \approx \GT_i(x)$  for every $x$.
	\item Even with access to $\Oracle(x)$, the function $\parity(x)$ remains hard to approximate. That is, for every polynomial $p$ of degree at most $n / \polylog(n)$, we have that $p(\Oracle(x))$ fails to approximate $\parity(x)$.
\end{itemize}
We can now obtain our improved lower bound by applying Buhrman and de Wolf's argument to the redundantly encoded inputs. Specifically, given a robust polynomial $p : \zbits^N \to \mathbb{R}$ of degree $d$ approximating $\OS_N$, we would have $p(q_{0}(\Oracle(x)), \dots, q_{N-1}(\Oracle(x))) \approx \OS_N(\GT_0(x), \dots, \GT_{N-1}(x)) = \parity(x)$ for every $x$. Our upper bound on the degrees of the $q_i$'s, together with our lower bound on the degree needed to approximate $\parity$, imply that $d \polylog n \ge n / \polylog n$, and hence $d \ge \widetilde{\Omega}(n)$.

All that remains is to construct the appropriate encoding $\Oracle$. Our approach is inspired by Nisan's classic randomized \emph{communication} protocol for computing the two-party greater-than function.
The most helpful way to think about this protocol for our purposes is as follows. Suppose Alice and Bob hold strings $a, b \in \{0, 1\}^n$ and their goal is to determine whether the natural number represented by $a$ is at least that represented by $b$. They may do so by performing binary search to identify the minimum index $j$ for which $a_j \ne b_j$, at which point the answer is determined by which of $a_j$ or $b_j$ is $1$. Each step of this binary search can be conducted by testing the equality of a substring of $a$ with a substring of $b$. Each equality test, in turn, may be performed (with high success probability) by comparing the inner products of $a$ and $b$ with a shared random string. The protocol requires $\log n$ steps of binary search, and each equality test should be repeated $O(\log \log n)$ times to achieve high success probability, giving an overall communication cost of $\widetilde{O}(\log n)$.

Now let us see how to turn this communication protocol into a polynomial approximating $\GT_i(x)$. Think of $x$ as Bob's input to the communication protocol, and of Bob's role as passively computing an encoding $\Oracle(x)$ that consists of many inner products of $x$ with random strings. Now thinking of $i$ as Alice's input, she can compute $\GT_i(x)$ (with high probability) by repeatedly querying $\Oracle(x)$ at the locations that correspond to the appropriate inner products from the protocol described above. This results in a $\widetilde{O}(\log n)$ randomized query algorithm for computing $\GT_i(x)$ from $\Oracle(x)$, the success probability of which is a degree-$\widetilde{O}(\log n)$ polynomial in $\Oracle(x)$.

The final step is to argue that even given $\Oracle(x)$, consisting of many inner products of random strings with $x$, the parity function $\parity(x)$ remains hard to compute. To see why this is true, note that a single inner product of $x$ with a random bit string is itself a parity on a random subset of indices. That is, $\Oracle(x) = (\parity(x|_{S_1}), \dots, \parity(x|_{S_m}))$ for random subsets $S_1, \dots, S_m \subseteq [n]$. The key observation then, is that a degree-$d$ polynomial of these random parities is able to approximate the full $\parity(x)$ if and only if some degree-$d$ polynomial of these random parities \emph{exactly} computes $\parity(x)$, which in turn happens if and only if a symmetric difference of at most $d$ of the sets $S_1, \dots, S_m$ yields the entire set of indices $[n]$. 
As a result, as long as neither the degree $d$ nor the number of random inner products $m$ is too large, we obtain that $\parity(x)$ cannot be approximated using $\Oracle(x)$.\footnote{In fact, this argument shows that it is impossible to approximate $\parity(x)$ to bounded error, but even to represent it in sign. This corresponds to a \emph{threshold degree} lower bound.}


\subsection{Our results in detail}

Recall that we introduce a framework that allows us to prove lower bounds on approximate degree, and hence quantum query complexity. It most naturally applies to decision versions of oracle identification problems, and extends to the ``weakly unbounded error'' setting of error approaching $1/2$. We summarize the results we prove using this framework in Table~\ref{tbl:results}.


\begin{table}[]
\begin{tabular}{|l|l|l|l|l|}
\hline
Problem                                                                     & Model                                                                                                                            & Error    & Previous work                                                                                                          & This work                                                                     \\ \hline
                                                                            &                                                                                                                                  & Unbounded & \begin{tabular}[c]{@{}l@{}}$O(n - \log{\frac{1}{\gamma}})$,\\ $\Omega(\sqrt{n} - \log{\frac{1}{\gamma}})$\end{tabular} & {\color[HTML]{FE0000} $\Omega(\frac{n}{\log^2{n}} - \log{\frac{1}{\gamma}})$} \\ \cline{3-5} 
                                                                            & \multirow{-3}{*}{\begin{tabular}[c]{@{}l@{}}Approximate degree and \\ quantum query complexity,\\ decision version\end{tabular}} & Constant  & $O(n)$, $\Omega(\sqrt{n})$                                                                                             & {\color[HTML]{FE0000} $\Omega(\frac{n}{\log^2{n}})$}                          \\ \cline{2-5} 
                                                                            &                                                                                                                                  & Unbounded & $\Theta(n - \log{\frac{1}{\gamma}})$                                                                                   & {\color[HTML]{333333} $\Omega(\frac{n}{\log^2{n}} - \log{\frac{1}{\gamma}})$} \\ \cline{3-5} 
\multirow{-4}{*}{\begin{tabular}[c]{@{}l@{}}Ordered \\ search\end{tabular}} & \multirow{-2}{*}{\begin{tabular}[c]{@{}l@{}}Quantum query complexity,\\ reconstruction version\end{tabular}}                     & Constant  & $\Theta(n)$                                                                                                            & $\Omega(\frac{n}{\log^2{n}})$                                                 \\ \hline
                                                                            &                                                                                           \multirow{-1.5}{*}{\begin{tabular}[c]{@{}l@{}}Approximate degree and \\ quantum query complexity,\\ decision version\end{tabular}}                                      & \begin{tabular}[c]{@{}l@{}}\vspace{-2.5mm}\\ Unbounded  \\ \vspace{-2.5mm} \end{tabular} & $O(n - \log{\frac{1}{\gamma}})$                                                                                        & {\color[HTML]{FE0000} $\Omega(\frac{n}{\log^2{n}} - \log{\frac{1}{\gamma}})$} \\ \cline{3-5} 
                                                                            &  & Constant   & $O(n)$                                                                                                                 & {\color[HTML]{FE0000} $\Omega(\frac{n}{\log^2{n}})$}                          \\ \cline{2-5} 
                                                                            &                                                                                                                                  & Unbounded  & \begin{tabular}[c]{@{}l@{}}$O(n - \log{\frac{1}{\gamma}})$,\\ $\Omega(\gamma^2 \frac{n}{\log^2{n}})$\end{tabular}      & {\color[HTML]{FE0000} $\Omega(\frac{n}{\log^2{n}} - \log{\frac{1}{\gamma}})$} \\ \cline{3-5} 
\multirow{-6}{*}{\begin{tabular}[c]{@{}l@{}}Hidden \\ string\end{tabular}}  & \multirow{-3}{*}{\begin{tabular}[c]{@{}l@{}}Quantum query complexity,\\ reconstruction version\end{tabular}}                     & Constant   & $O(n)$                                                                                                                 & $\Omega(\frac{n}{\log^2{n}})$                                                 \\ \hline
\end{tabular}
\caption{Summary of our results and prior work.}
\label{tbl:results}
\end{table}

\paragraph{Ordered search.} 
As mentioned, binary search yields a deterministic algorithm making $n$ queries, which in turn yields a polynomial of degree $n$ that exactly computes $\OS_{2^n}$. To compute this function with error probability $\frac{1}{2} - \gamma$ for some parameter $\gamma > 0$, there is an easy way to modify binary search to obtain an $O(n - \log{\frac{1}{\gamma}})$-query randomized algorithm (see Appendix \ref{sec:unbounded-ub} for details). This implies an upper bound of $O(n - \log{\frac{1}{\gamma}})$ 
 on the approximate degree of $\OS_{2^n}$ with error parameter $1/2 - \gamma$.

Before this work, the best lower bound on approximate degree (for both bounded and unbounded error) was obtained by \cite{BuhrmanW99} and was $\Omega(\sqrt{n} - \log{\frac{1}{\gamma}})$ for approximation to error $\frac{1}{2} - \gamma$. We significantly improve their result and obtain the following lower bound.

\begin{theorem*}
For every natural number $n$ and $0 < \gamma < 1/2$, every polynomial that approximates $\OS_{2^n}$ pointwise to error $\frac{1}{2} - \gamma$ requires degree
$$\Omega\left(\frac{n}{\log^2{n}} - \log{\frac{1}{\gamma}}\right).$$
\end{theorem*}
This result is restated as Theorem \ref{thm:os2-lb}. It shows that it is hard to approximate the decision version of the ordered search problem $\OS_{2^n}$ (with $\parity$ as the predicate converting from reconstruction to decision problem) not only to constant error, but even to small advantage $\gamma$ over random guessing. For instance, approximating $\OS_{2^n}$ with advantage $\gamma = 2^{-n^{0.99}}$ still requires degree $\Omega(\frac{n}{\log^2{n}})$. Our lower bound is nearly tight in both the bounded and unbounded error regimes.

\paragraph{Query complexity of ordered search.} 
Most previous work on the quantum query complexity of ordered search addressed the bounded error regime and the reconstruction version of the problem, where the goal is to output the entire string $x$, rather than a specific Boolean predicate applied to $x$. To our knowledge, the best prior lower bound for the decision version of ordered search with unbounded error follows from \cite{BuhrmanW99} as described above and is $\Omega(\sqrt{n} - \log{\frac{1}{\gamma}})$. Note also that the $\Omega(n)$ lower bound of \cite{Ambainis99}, stated there for constant error, also generalizes to a tight lower bound $\Omega(n - \log{\frac{1}{\gamma}})$ for unbounded error, but it appears to hold only for the reconstruction version of ordered search.

Our application of the polynomial method implies a nearly tight quantum query lower bound that applies to the decision version of the problem.


\begin{corollary} \label{cor:os-ue-lb}
    Every quantum algorithm that computes $\OS_{2^n}$ (decision version with $\parity$) with probability of error at most $\frac{1}{2} - \gamma$ requires $\Omega(\frac{n}{\log^2{n}} - \log{\frac{1}{\gamma}})$ queries.
\end{corollary}

\paragraph{Hidden string.} 
The work of \cite{SkienaS95} yields a simple deterministic algorithm making $O(n)$ queries, which in turn yields a polynomial of degree $O(n)$ that exactly computes $\HS_{2^{n+1} - 1}(\dots, \phi_s(x), \dots) := \parity(x)$ where $x \in \{0, 1\}^n$ is the hidden string in question. Again, this algorithm can be modified to get a $O(n - \log{\frac{1}{\gamma}})$-query algorithm with error $1/2-\gamma$ (see Appendix \ref{sec:unbounded-ub} for details). This implies an upper bound $O(n - \log{\frac{1}{\gamma}})$ on the approximate degree of $\HS_{2^{n+1}-1}$.

We give the first lower bound on the approximate degree of the hidden string problem:

\begin{theorem*} 
For every natural number $n$ and $0 < \gamma < 1/2$, every polynomial that approximates $\HS_{2^{n+1}-1}$ to error $\frac{1}{2} - \gamma$ requires degree

$$\Omega\left(\frac{n}{\log^2{n}} - \log{\frac{1}{\gamma}}\right).$$

\end{theorem*}

This result is restated as Corollary \ref{cor:hs-lb}, and gives a nearly tight lower bound for approximating the decision version of $\HS_{2^{n+1} - 1}$ to both constant and weakly unbounded error. 

\paragraph{Query complexity of hidden string.} 
Complementing the $O(n)$-query deterministic algorithm of~\cite{SkienaS95}, it turns out that a constant-factor speedup is possible for quantum algorithms \cite{cleve2012reconstructing}.  As for lower bounds, the latter work shows a lower bound $\Omega(\frac{n}{\log^2{n}})$ on reconstruction by adversary method. This lower bounds holds for bounded error, but does not generalize well to unbounded error regime. (By \cite{barnum2001quantum,hoyer2007negative}, the same proof implies a lower bound of $\Omega(\gamma^2 \frac{n}{\log^2{n}})$ for solving the reconstruction version of hidden string with error $\frac{1}{2} - \gamma$.)

Our approximate degree lower bound recovers their lower bound for bounded error, and gives a significantly stronger lower bound for the weakly unbounded error regime, both for the decision version of the problem.
 
\begin{corollary} \label{cor:hs-ue-lb} 
    Every quantum algorithm that computes $\HS_{2^{n+1}-1}$ (decision version with $\parity$) with probability of error at most $\frac{1}{2} - \gamma$ requires $\Omega(\frac{n}{\log^2{n}} - \log{\frac{1}{\gamma}})$ queries.
\end{corollary}

\subsection{Further discussion}

One of our initial motivations for studying the approximate degree of ordered search came from the preliminary version of Chattopadhyay et al.~\cite{ChattopadhyayKLM17}. They showed that $\OS_N \circ \IP_m^N$ has randomized communication complexity $\Omega(\log N \cdot m)$, where $\IP_m$ is a two-party inner product (mod 2) gadget on $m$-bit inputs. This was done via an involved simulation argument, showing how a communication protocol for $\OS_N \circ \IP_m^N$ could be used to construct a randomized decision tree for $\OS_N$. The techniques were specialized to the both the outer function and the inner function. Subsequent work~\cite{ChattopadhyayFKMP21} recovered this result using a generic simulation theorem. A direct application of Sherstov's pattern matrix method~\cite{Sherstov11Pattern}  to our result yields a \emph{quantum} communication lower bound of $\Omega(\log N / \log^2 \log N)$ on $\OS_N \circ g^N$ even for a constant-sized gadget $g$.

Hoza~\cite{Hoza17} used ideas conceptually related to ours to nearly recover the known quantum query (but not approximate degree) lower bound for ordered search. Roughly, he used a Holevo-information argument to show that if an oracle identification problem specified by functions $a_1, \dots, a_N$ can be solved with $T$ quantum queries, then $Q^*(A) \cdot T \gtrsim n$, where $A(i, x) = a_i(x)$ and $Q^*$ is the bounded-error two-party quantum communication complexity with shared entanglement. His quantum query lower bound for ordered search follows directly from the fact that the quantum communication complexity of the two-party greater-than function $\GT$  on $n$-bit inputs is $O(\log n)$. However, without opening up the communication protocol for $\GT$ as we do, it is not clear how to recover an approximate degree lower bound from his construction.

The idea of indirectly proving approximate degree lower bounds by combining a lower bound for one problem with an upper bound for another also appears in~\cite{Ben-DavidBGK18}. They gave a tight lower bound on the approximate degree of any function of the form $f \circ g^n$ where $f$ is an $n$-input symmetric function by combining a known lower bound for $\parity \circ g^n$~\cite{Sherstov12} with a quantum query and approximate degree upper bound for the combinatorial group testing problem~\cite{Belovs15}.

We believe it should be possible to extend our techniques to prove new lower bounds for other oracle identification problems. A family of special cases of oracle identification is captured by the symmetric junta learning problem~\cite{AmbainisM14}. Here, there is a symmetric function $h : \zbits^k \to \zbits$ and each $f_S$ takes the form $f_S(x) = h(x |_S)$. An important instance of this problem is the combinatorial group testing problem, wherein one takes $h = \OR_k$. Belovs gave a tight upper bound of $O(\sqrt{k})$~\cite{Belovs15} for this problem. He also determined the query complexity for $h = \mathsf{EXACT}-\mathsf{HALF}$ to be $\Theta(k^{1/4})$ and gave an upper bound of $O(k^{1/4})$ for $h = \MAJ$. These upper bounds were also (nearly) recovered algorithmically by Montanaro and Shao~\cite{MontanaroS20}. Despite its similarity to $\mathsf{EXACT}-\mathsf{HALF}$, no polynomial lower bound is known for the majority function $\MAJ$.

In the counterfeit coin problem, there is a hidden string $x \in \bits^n$ with Hamming weight at most $k$. A query is parameterized by a balanced (i.e., having an equal number of $1$'s and $-1$'s) string $y \in \{-1, 0, 1\}^n$, and indicates whether $\langle x, y \rangle$ is zero or non-zero. Iwama et al.~\cite{IwamaNRT12} gave a quantum algorithm making $O(k^{1/4})$ queries and conjectured this is tight, but no lower bound is known. Note that the oracle here is quite similar to the $\mathsf{EXACT}-\mathsf{HALF}$ oracle.



\section{Technical ideas}\label{sec:roadmap}

\subsection{Our lower bound framework}

We begin with a somewhat more abstract description of our framework for proving approximate degree lower bounds for oracle identification problems. The main idea is to provide additional information about the hidden input to an oracle identification problem so as to selectively affect the ability of quantum query algorithms and approximating polynomials to compute the functions we wish to understand.



Recall that an oracle identification problem is specified by a family of functions $a_1, \dots, a_N$. Given query access to the values $a_1(x), \dots, a_N(x)$, the goal in our decision problems is to compute the function $\parity(x)$. Suppose that we may identify $\parity(x) = f(a_1(x), \dots, a_N(x))$ for some function $f$. If we can construct a function $\Oracle$ such that:



\begin{itemize}
    \item Given $\Oracle(x)$, every function $a_i(x)$ can be computed by a low-degree polynomial, but
    \item Given $\Oracle(x)$, computing the parity of $x$ requires a high-degree polynomial,
\end{itemize}

Then by combining these two statements, we see that the function $f(a_1, \dots, a_N)$ itself requires a high-degree polynomial.
We apply this framework taking $f$ to be either the $\OS$ function or for the ``anchored hidden string'' $\AHS$ function. The latter also implies a lower bound for the original (decisional) hidden string function $\HS$ described in the introduction. 

In the following sections, we describe the main technical ideas that go into the proofs of our lower bounds. In order to provide more intuition about the structure of $\Oracle$, we describe the steps of constructing it for $\OS$ in detail before returning to the generalized framework.

\subsection{Ordered search lower bound}

First, notice that $\OS_N$ has the structure of an oracle identification problem since 
$$\OS_N(\GT_{0^n}(x), \GT_{0^{n-1}1}(x), \ldots, \GT_{1^n}(x)) = \parity(x)$$
where $N = 2^n$ and $\GT_i(x) = 1$ if and only if $x \leq i$ where $i, x \in \{0, 1\}^n$ if compared as numbers written in binary notation.



We want to show that there exists a function $\Oracle$ of $x$ that we think of as revealing partial information about $x$ such that:

\begin{itemize}
    \item On one hand, for all $i \in \{0, 1\}^n$ there is an algorithm that makes a small number of queries to $\Oracle$ and can identify the value of $\GT_i(x)$ with constant probability of success. Note that a query-efficient algorithm automatically gives rise to a low-degree approximating polynomial.
    \item On the other hand, approximating the value of $\parity(x)$ given $\Oracle$ with any probability of success requires a lot of queries to $\Oracle$. Let us denote this auxiliary problem by $\CRAP(\Oracle) := \parity(x)$.
\end{itemize}
It is helpful to think of $\Oracle$ itself as an oracle, whose output is given to a polynomial or to a query algorithm, whose goal is then to compute some other function of $x$. We describe how we construct oracle $\Oracle$ through several attempts. 

Let us first focus on constructing an oracle $\Oracle$ that meets the first condition. To do so, we can use the idea behind the $O(\log{n} \log{\log{n}})$-bit communication protocol\footnote{A more efficient $O(\log n)$-bit communication protocol is known and underlies our sharpest result for ordered search. We discuss it in Sections~\ref{sec:roadmap-improved} and~\ref{sec:OS}.} for the two-party communication problem $\GT$ to obtain an efficient randomized query algorithm for every function $\GT_i$. In the $\GT$ communication problem, Alice and Bob both get a string of $n$ bits and the goal is to decide if the number represented by Alice's string is greater than the number represented by Bob's string. 

In this randomized communication protocol for $\GT$, Alice checks if the first halves of the inputs are equal and depending on the answer, she either recursively continues on the first halves of the inputs or the second halves. By doing so, she finds the most significant bit where the inputs differ. To perform each equality check, both Alice and Bob compute the inner products modulo 2 of each of the inputs with the same set of some $\alpha$ (publicly) random strings, Bob sends his values to Alice, and Alice compares these values to the values she obtained. If the original values were equal, then the inner products will be always equal, and otherwise, at least one pair of inner products will be unequal with high probability for sufficiently large $\alpha$. This elementary operation (i.e., the ability to compute inner products with random strings) will be exactly what we want our oracle $\Oracle$ to be useful for.

\paragraph{First attempt.}
We will eventually give a randomized construction of the oracle $\Oracle$, and to this end, think of it as taking as input both the hidden string $x$ and a random input $r$. Let $\Oracle(r, x)$ be a function that takes a collection of $m$ $n$-bit strings $r \in \times_{i \in [m]} (\{0, 1\}^n)$ and $x \in \{0, 1\}^n$, and outputs $m$ bits, each representing the inner product of  $r_i$ with $x$: $(\Oracle(r, x))_i = \langle r_i, x\rangle$. 

Our first attempt, however, will make no use of randomness at all. Let us consider $\Oracle(r, x)$ where $r$ consists of all possible strings of length $n$. That is, the output of the oracle consists of $\langle x, r_i \rangle$ for every $r_i \in \{0, 1\}^n$.


Let us now see how to construct a query algorithm $C_i$ that, given oracle access to $\Oracle(r, x)$, computes $\GT_i(x)$ with high probability. This algorithm emulates Alice's side in the communication protocol, fixing her input to $i$. It samples random strings used in the communication protocol, computes the inner products of $i$ with these random strings on its own, and asks the oracle (emulating Bob) for the inner products of $x$ with the same random strings.

From the correctness of the communication protocol for $\GT$ we can conclude that for all $x, i \in \{0, 1\}^n$

$$\Pr_{{r_1\ldots, r_{\alpha \log{n}}}}[C_{i}(\Oracle(r, x)) \neq \GT_i(x)] < \log{n} \cdot 2^{-\alpha}$$
where ${r_1\ldots, r_{\alpha \log{n}}}$ are the strings that $C_{i}$ sampled during the run, and $r$ is a collection of all $n$-bit strings. The number of queries is $\alpha \log{n}$.

Thus we see that this oracle satisfies the first condition: it helps to compute the $\GT_i$ efficiently for every $i$ and $x$. But now there is a problem with the second condition: $\parity(x) = \CRAP(\Oracle)$ can be computed easily since $\parity(x) = \CRAP(\Oracle(r, x)) = \langle x, 1^n\rangle$. So there is a 1-query algorithm (and hence a degree-1 polynomial) that exactly  computes $\CRAP(\Oracle(r, x))$, violating our second condition.

\paragraph{Second attempt.} Our goal now is to reduce the efficacy of the oracle $\Oracle$ in terms of how well it can be used by low-degree polynomials to approximate $\CRAP$. To do this, we instead consider a distribution over the potential oracles defined by the collection of strings used in the protocol. Let $r$ denote a sequence of the random strings that could appear in one run of $\GT$ protocol described earlier. Let $\hat{\R}$ denote the set of all such sequences.
This allows us to define a distribution of oracles $\OracleWithParameter{\hat{\R}}(r, x)$, where $r \from \hat{\R}$, and for us to consider a deterministic query algorithm. Let $B_{(r, i)}$ be a deterministic algorithm that is given access to the $\OracleWithParameter{\hat{\R}}(r, x)$ where $r \from \hat{\R}$ is chosen uniformly at random, and which has the realization of $r$ and $i$ hardcoded into it. This algorithm is able to emulate the communication protocol (and the algorithm $C_{i}$), but now each time it needs a random string, it uses one provided in $r$.

From the correctness of the communication protocol for $\GT$ we again can conclude that for all $x, i \in \{0, 1\}^n$

$$\Pr_{r \from \hat{\R}}[B_{(r, i)}(\OracleWithParameter{\hat{\R}}(r, x)) \neq \GT_i(x)] < \log{n} \cdot  2^{-\alpha}.$$
So, with high probability, $B_{(r, i)}$ computes $\GT_i(x)$ over the choice of the oracle $\OracleWithParameter{\hat{\R}}(r, x)$ for $r \from \hat{\R}$. 

Does this new oracle satisfy the second condition? Now an approximation to $\CRAPWithParameter{\hat{\R}}(\OracleWithParameter{\hat{\R}}(r, x))$ needs to approximate $\parity(x)$ when given a set of random parities from $\hat{\R}$. Indeed, we show this requires high degree, as a consequence of the fact that high degree polynomial is necessary to construct the full parity of $x$ from random parities. 

However, we need to add one more improvement to our structure. For every fixed $i, x$, the algorithm $B_{(r, i)}$ when run on $\OracleWithParameter{\hat{\R}}(r, x)$ computes $\GT_i(x)$ with high probability over $r \from \hat{\R}$. But we need to switch quantifiers: we want an oracle that is ``good'' for all possible inputs for $\GT$ simultaneously and, unfortunately, our current construction doesn't give an algorithm computing $\GT_i(x)$ for all $i,x \in \{0, 1\}^n$ using the same $r \from \hat{\R}$.

\paragraph{Third (and final) attempt.} So, is there a way to fix the source of randomness so it works for all possible inputs? Inspired by Newman's theorem \cite{newman1991private} on simulating public randomness using private randomness in communication complexity, we show that there is. We show that by taking $t = O(\frac{n}{\delta^2})$ copies of $\hat{\R}$, denoted $\R_1, \R_2, \ldots \R_t$, we get a ``good base'' for the oracle. Consider a randomized algorithm $A_{(r, i)}$ that, given access to to $\OracleWithParameter{\R'}(r, x)$ with $r \from \R' = \times_{j \in [t]} \R_j$, does the following:

\begin{itemize}
    \item Sample $j \leftarrow [t]$ at random.
    \item Run $B_{(r, i)}$ using the set $\R_j$ as the source of randomness.
\end{itemize}

Following the argument underlying Newman's theorem, we show that this algorithm computes $\GT_i(x)$ with $\log{n} \cdot 2^{-\alpha} + \delta$ failure probability. It works for every $i$ and $x$ and it still makes only $\alpha \log{n}$ queries to the oracle. If we put $\delta = \frac{1}{12}$ and $\alpha = O(\log{\log{n}})$ then the probability of this algorithm failing for some input pair is at most $\frac{1}{6}$ with only $\alpha \log{n} = O(\log{n} \log{\log{n}})$ queries to the oracle, i.e.,

$$\Pr_{r \from \R'}[A_{(r, i)}(\OracleWithParameter{\R'}(r, x)) \neq \GT_i(x)] < \frac{1}{6}.$$

This change also doesn't increase the ``size'' of the oracle (i.e., the number of queries it can answer) too much. This allows us to show that with high probability it is still impossible to combine the given partial parities to create the full parity using a low-degree polynomial, so the second condition is also satisfied. So there exists an oracle that allows computing the $\GT$ with low-degree polynomials but requires a high-degree polynomial to compute $\parity(x)$ which is exactly what allows us to prove the lower bound on the approximate degree of $\OS$.

\subsection{Technical ideas behind the parity lower bound}

Our technique relies on a lower bound on the approximate degree of $\parity(x)$, or, more precisely, on the ``Parity Under Randomness $\R$'' function $\CRAPWithParameter{\R}(\OracleWithParameter{\R}(r, x))$ evaluates to $\parity(x)$ on input $\OracleWithParameter{\R}(r, x)$.   
We, in fact, prove a more general statement lower bounding the approximate degree of $\CRAPWithParameter{\R}$ for a class of potential structures $\R$.

Specifically, we show that the parity function is hard, even to sign-represent, and even given access to $\OracleWithParameter{\R}$ consisting of inner products of $x$ with random strings $r_i$ where each bit of $r_i$ is either fixed to zero or is an unbiased random bit. The only other restriction we need on $\OracleWithParameter{\R}$ is that its ``size'', i.e., the number of inner products it provides, is small. The bigger this number is, the worse our the lower bound becomes.

The proof idea is based on the hardness of sign-representing parity as described in \cite{aspnes1991expressive}, combined with the following combinatorial observation: given a set of $n$-bit strings (corresponding to samples from $\R$, and in turn to random inner products) where in every string each bit is either zero or is an unbiased random bit, with high probability no small subset of them adds up to the all-ones string (which corresponds to the $\parity$ function).

\subsection{Improved ordered search and anchored hidden string lower bounds} \label{sec:roadmap-improved}

Our generalized lower bound for approximating $\CRAPWithParameter{\R}$ allows us to obtain other lower bounds for oracle identification problems. 
For example, we give a slightly stronger lower bound for $\OS$ than what is implied by the discussion above. There is, in fact, a more efficient randomized communication protocol for $\GT$ that uses $O(\log{n})$ bits of communication. It can be converted into randomized query algorithm and thus into a polynomial of degree $O(\log{n})$. At the same time, this more efficient protocol is still based on computing equalities of substrings of inputs, and so the appropriate $\Oracle$ has a very similar structure to the one described above while still satisfying the conditions of the generalized lower bound for $\CRAP$. Moreover, the necessary ``size'' of $\Oracle$ barely blows up at all. Putting everything together gives our improved lower bound of $\Omega\left(\frac{n}{\log^2 n}\right)$ on the approximate degree of $\OS$.

Using the same framework, we can also obtain a nearly tight lower bound on the approximate degree of the anchored hidden string problem $\AHS$. In the anchored hidden string problem, the goal is to determine the parity of $x$ given oracle access to $y_{i, s} = \phi_{i, s}(x)$ for every index $i$ and every binary string $s$ of length at most $n$, where $\phi_{i, s}(x) = 1$ iff the substring of $x$ starting at index $i$ matches $s$. This oracle identification problem has the right form for our framework since

$$\AHS_N((\phi_{i, s}(x))_{i \in [n], s \in \{0, 1\}^{\leq n-i+1}}) = \parity(x).$$
Moreover, each function $\phi_{i, s}(x)$ simply computes the equality function of $s$ with a substring of $x$ of length $|s|$ starting from position $i$. As we have already seen, we can compute the equality function very efficiently given an oracle $\Oracle$ of the right random structure, and such a $\Oracle$ meets the conditions of our generalized lower bound for $\CRAPWithParameter{\Oracle}$. This directly implies a lower bound of $\Omega\left(\frac{n}{\log n}\right)$ on the approximate degree of $\AHS$.

Finally, the last lower bound described in this work is on the approximate degree of $\HS$. This lower bound follows via a reduction from $\AHS$. This reduction was first introduced in \cite{cleve2012reconstructing} in the quantum query model, but it holds for polynomial approximation as well.

\section{Ordered search and generalized lower bound}\label{sec:OS}

In this section we give the formal proof of our lower bound on the approximate degree of ordered search. We show how our framework is used for this function and prove the generalized lower bound on $\parity$ that we later reuse for the hidden string problem.

\subsection{Preliminaries}

Our lower bounds on the approximate degree of (a decision version) of ordered search and the hidden string problem require the following definition of polynomial approximations for promise problems.

\begin{definition}
Let $f : D \to \{0, 1\}$ where $D \subseteq \{0, 1\}^n$ for some $n \in \mathbb{N}$ be a partial Boolean function. For $\frac{1}{2} > \eps > 0$, a polynomial $p : \{0, 1\}^n \to \mathbb{R}$ is an $\eps$-approximation to $f$ if $|p(x) - f(x)| \le \eps$ for every $x \in D$ and $-\eps \le p(x) \le 1 + \eps$ for all $x \in \{0, 1\}^n$. The $\eps$-approximate degree of $f$, denoted $\adeg_\eps(f)$ is the the least degree of a polynomial $p$ that $\eps$-approximates $f$. We use the convention $\adeg(f) = \adeg_{1/3}(f)$ to refer to the ``approximate degree of $f$'' without qualification.
\end{definition}

That is, we require a polynomial approximation to a partial function defined on a domain $D$ to approximate the function on $D$ and remain bounded outside of $D$. Note that this is the type of approximation that arises from quantum query algorithms for promise problems.

We also formally define the ordered search function $\OS$ and the family of greater-than functions $\GT$. 

\begin{definition}
    For all $i \in \{0, 1\}^n$ define the function $\GT_i: \{0, 1\}^n \to \{0, 1\}$ to be the indicator of whether the value of the input is smaller than $i$: $\GT_i(x) = 1$ if and only if $x \leq i$ where $i$ and $x$ are compared as numbers written in binary notation.
\end{definition}

\begin{definition}
    The ordered search function $\OS_{2^n}: \{0^{k}1^{2^n-k}\mid k \in [2^n]\} \to \{0, 1\}$ is a partial function defined the following way: $\OS_{2^n}(0^{k}1^{2^n-k}) = \parity(x)$ where $x \in \{0, 1\}^n$ is the binary representation of $k$.
\end{definition}

\subsection{The notion of a \emph{good base}.}

In order to formally define the oracle, i.e. the source of additional information about the input, we introduce the notion of a ``\emph{good base}''  for the oracle. A set $\R$, consisting of tuples of strings, is a \emph{good base} if it's constructed as follows. 

Let $\R'$ be a Cartesian product of $m'$ subsets of $\{0, 1\}^{n}$ where each subset $\R^{\tau}$ is itself defined by an $n$-bit string-template $\tau = \tau_1 \tau_2 \ldots \tau_n \in \{0, 1\}^{n}$

$$\R^{\tau} = S_{\tau_{1}} S_{\tau_{2}} S_{\tau_{3}} \ldots S_{\tau_{n}}$$ 
where $S_0 = \{0\}$ and $S_1 = \{0, 1\}$. 

For example, if $\tau = 00100010$ then $\R^{\tau} = S_{\tau_{1}} S_{\tau_{2}} S_{\tau_{3}} \ldots S_{\tau_{n}} = S_0 S_0 S_1 S_0 S_0 S_0 S_1 S_0 = $ \\$\{0\} \{0\} \{0, 1\} \{0\} \{0\} \{0\} \{0, 1\} \{0\} = \{00000000, 00000010, 00100000, 00100010\}$.

Let $\mathcal{B} = \{\one_1\} \times \{\one_2\} \times \ldots \times \{\one_n\}$ where $\one_j = 0^{i-j}10^{n-j}$ is the string that has the value 1 in $j$-th position and has the value 0 everywhere else. Let $\R =  \mathcal{B} \times \R'$, and thus $\R$ is a Cartesian product of $m = n + m'$ subsets of $\{0, 1\}^n$. Note that every $r \in \R$ is a $m$-tuple of $n$-bit strings:
$$r = (r_1, r_2, \ldots r_m) = (\one_1, \one_2, \ldots, \one_{n-1}, \one_n, r_{n+1}, r_{n+2}, \ldots, r_{n+m'})$$
where each $r_j$ is a string of length $n$, the first $n$ strings are fixed for all $r \in \R$, and the last $m'$ strings are from some sets $\R^{\tau}$ each for some template $\tau$. If $r \from \R$ is chosen u.a.r. then each $r_j,n < j \leq m'$ is chosen u.a.r. from some $\R^{\tau}$ and thus the subsequence of bits of $r_j$ corresponding to ones in $\tau$ is a uniformly random string, and the subsequence of bits of $r_j$ corresponding to zeros in $\tau$ is the all-zero string.

Any set $\R$ with the above structure will be called a \emph{good base} of size $m$. Such an $\R$ is helpful for building our oracles as follows.

Let $\OracleWithParameter{\R}: \R \times \{0, 1\}^n \to \{0, 1\}^{m}$ be the following function: $(\OracleWithParameter{\R}(r, x))_j = \langle r_j, x\rangle$ where $r_j$ is an $n$-bit string from the collection $r \in \R$ and the inner product is taken modulo 2. Note that $\OracleWithParameter{\R}$ is parameterized by $\R$, so for each \emph{good base} $\R$ the function $\OracleWithParameter{\R}$ will be different. We will omit the parameter $\R$ later in places where it is clear from context.

Notice the following properties of this function $\OracleWithParameter{\R}(r, x)$ that hold whenever $\R$ is a \emph{good base}:

\begin{itemize}
    \item For every $r \in \R$, the values $\Oracle(r, x)$ completely determine $x$. Since the first $n$ strings of $r$ are $\one_1$, $\one_2$, $\ldots$, $\one_{n-1}$, $\one_n$, the first $n$ bits of $\Oracle(r, x)$ are exactly bits of $x$.
    \item Given $\Oracle(r, x)$ for $r \from \R$ and $r$ itself, one can compute (with some probability of error) whether a subsequence of $x$ specified by some pattern $\tau$ agrees with some fixed string $s$ in those indices. To be more specific, if given $(\Oracle(r, x))_j = \langle r_j, x \rangle$ where $r_j$ is sampled from $\R^{\tau}$ uniformly at random, and $r_j$ itself, one can check whether the strings $x \land \tau$ (where $\land$ denotes bitwise AND)  and $s \land \tau$ are equal for any $s \in \{0, 1\}^n$ with one-sided error probability $\frac{1}{2}$.
\end{itemize}

So, $\OracleWithParameter{\R}(r, x)$ could be used as an equality oracle for a fixed set of subsequences of $x$ (predefined by $\R$) when $r$ is chosen uniformly at random from $\R$. Thus, $\OracleWithParameter{\R}(r, x)$ might give more information about $x$ than $x$ alone and might make some computations on $x$ more efficient. 

On the other hand, some functions of $x$ remain ``hard'' even when given $\Oracle(r, x)$. We will later show that $\parity(x)$ remains hard to compute even with this additional information.

\subsection{Approximating polynomials for $\GT_i$}

We start our proof by showing that for some \emph{good base} $\R_{\OS}$ the oracle $\OracleWithParameter{\R_{\OS}}$ could be used to make the computation of $\GT$ functions more efficient.

\begin{claim}\label{clm:os-ub}
   There exists a \emph{good base} $\R_{\OS}$ of size $m = O(n^2 \log{\log{n}})$ such that if $r \from \R_{\OS}$ is sampled uniformly at random, then with probability at least $\frac{2}{3}$ over the choice of $r$ there exists a family of $2^n$ polynomials $\{q_{(r, i)}: \{0, 1\}^{m} \to \{0, 1\}\mid i \in \{0, 1\}^n\}$, each of degree at most $2 \log{n} \log{\log{n}}$, such that given $\Oracle[\R_{\OS}](r, x)$ as the input, each polynomial $q_{(r, i)}(\OracleWithParameter{\R_{\OS}}(r, x))$ approximates the corresponding $\GT_i(x)$ with error at most $\frac{1}{6}$. That is,
   
    $$\Pr_{r \from \R_{\OS}}\left[\exists i, x \in \{0, 1\}^n: \left|q_{(r, i)}(\Oracle(r, x)) - \GT_i(x)\right| > \frac{1}{6}\right] < \frac{1}{3}.$$
    
\end{claim}
\begin{proof} 

This proof consists of two parts: constructing a \emph{good base} $\R_{\OS}$ and showing that it actually helps to compute every $\GT_i$.

\textbf{Constructing the \emph{good base} $\R_{\OS}$.} We are going to construct $\R_{\OS}$ based on what random strings are useful in the communication protocol computing $\GT$ of two $n$-bit strings, $x$ and $i$. Intuitively, in this protocol, we first need to check if the first half of $i$ and $x$ are equal using a randomized communication protocol for equality. To do that we need to compute and compare $\langle x, r\rangle$ and $\langle i, r\rangle$, for some number $\alpha$ of random strings $r$ to be determined later, where each $r$ is sampled from $\{0, 1\}^{\frac{n}{2}} \{0\}^{\frac{n}{2}}$. If the computed values $\langle i, r\rangle = \langle x, r \rangle$ for all $r$ we have considered, then we repeat this procedure on the second half of $x$ and $i$, which corresponds to computing and comparing $\langle x, r\rangle$ and $\langle i, r\rangle$ for $\alpha$ random strings $r$ sampled from $\{0\}^{\frac{n}{2}}\{0, 1\}^{\frac{n}{4}} \{0\}^{\frac{n}{4}}$. If, on the other hand, the values were not equal then we repeat this procedure on the first half of $x$ and $i$, which corresponds to computing and comparing $\langle x, r\rangle$ and $\langle i, r\rangle$ for $\alpha$ random strings $r$ sampled from $\{0, 1\}^{\frac{n}{4}} \{0\}^{\frac{3n}{4}}$. Since we want our oracle to be useful to emulate this procedure to compute $\GT_i(x)$, it should ``contain'' all the random strings used in this protocol. 

Let $\hat{\R} = \R^{1^{n/2}0^{n/2}} \times \left(\R^{1^{n/4}0^{3n/4}} \times \R^{0^{n/2} 1^{n/4}0^{n/4}}\right) \times \ldots \times \left(\bigtimes_{i = 0}^{2^{k}-1}\R^{0^{2in/2^{k+1}} 1^{n/2^{k+1}}0^{n-((2i+1)n/2^{k+1})}}\right) \times \ldots \times \left(\bigtimes_{i = 0}^{n/2}\R^{0^{2i} 1^{1}0^{n-(2i+1)}}\right)$. See Figure~\ref{fig:struct-os} for an illustration.

\begin{figure}[h]
\begin{center}
\begin{tabular}{ccc}
{\Large $\tau$} & {\Large$\R^{\tau}$} & Structure of $\R^{\tau}$\\ \\
{\Large $1^{\frac{n}{2}}0^{\frac{n}{2}}$} & {\Large \hspace{1cm} $\{0, 1\}^{\frac{n}{2}}\{0\}^{\frac{n}{2}}$ \hspace{1cm}} & 
    \small
    \tabcolsep=0.12cm
    \begin{tabular}{|l|l|l|l|l|l|l|l|l|l|l|l|l|l|l|l|}
    \hline
   \cellcolor{blue_color}$\star$ &\cellcolor{blue_color}$\star$ &\cellcolor{blue_color}$\star$ &\cellcolor{blue_color}$\star$ &\cellcolor{blue_color}$\star$ &\cellcolor{blue_color}$\star$ &\cellcolor{blue_color}$\star$ &\cellcolor{blue_color}$\star$ & 0 & 0 & 0 & 0 & 0 & 0 & 0 & 0 \\ \hline
    \end{tabular}
    \normalsize
    \\ 
    \\
{\Large $1^{\frac{n}{4}}0^{\frac{3n}{4}}$} & {\Large $\{0, 1\}^{\frac{n}{4}}\{0\}^{\frac{3n}{4}}$} & 
    \small
    \tabcolsep=0.12cm
    \begin{tabular}{|l|l|l|l|l|l|l|l|l|l|l|l|l|l|l|l|}
    \hline
   \cellcolor{blue_color}$\star$ &\cellcolor{blue_color}$\star$ &\cellcolor{blue_color}$\star$ &\cellcolor{blue_color}$\star$ & 0 & 0 & 0 & 0 & 0 & 0 & 0 & 0 & 0 & 0 & 0 & 0 \\ \hline
    \end{tabular}
    \normalsize
    \\
{\Large $0^{\frac{n}{2}}1^{\frac{n}{4}}0^{\frac{n}{4}}$} & {\Large $\{0\}^{\frac{n}{2}}\{0, 1\}^{\frac{n}{4}}\{0\}^{\frac{n}{4}}$} & 
    \small
    \tabcolsep=0.12cm
    \begin{tabular}{|l|l|l|l|l|l|l|l|l|l|l|l|l|l|l|l|}
    \hline
    0 & 0 & 0 & 0 & 0 & 0 & 0 & 0 &\cellcolor{blue_color}$\star$ &\cellcolor{blue_color}$\star$ &\cellcolor{blue_color}$\star$ &\cellcolor{blue_color}$\star$ & 0 & 0 & 0 & 0 \\ \hline
    \end{tabular}
    \normalsize
 \\ 
 \\
{\Large $1^{\frac{n}{8}}0^{\frac{7n}{8}}$} & {\Large $\{0, 1\}^{\frac{n}{8}}\{0\}^{\frac{7n}{8}}$} & 
    \small
    \tabcolsep=0.12cm
    \begin{tabular}{|l|l|l|l|l|l|l|l|l|l|l|l|l|l|l|l|}
    \hline
   \cellcolor{blue_color}$\star$ &\cellcolor{blue_color}$\star$ & 0 & 0 & 0 & 0 & 0 & 0 & 0 & 0 & 0 & 0 & 0 & 0 & 0 & 0 \\ \hline
    \end{tabular}
    \normalsize\\
{\Large $0^{\frac{n}{4}}1^{\frac{n}{8}}0^{\frac{5n}{8}}$} & {\Large $\{0\}^{\frac{n}{4}}\{0, 1\}^{\frac{n}{8}}\{0\}^{\frac{5n}{8}}$} & 
    \small
    \tabcolsep=0.12cm
    \begin{tabular}{|l|l|l|l|l|l|l|l|l|l|l|l|l|l|l|l|}
    \hline
    0 & 0 & 0 & 0 &\cellcolor{blue_color}$\star$ &\cellcolor{blue_color}$\star$ & 0 & 0 & 0 & 0 & 0 & 0 & 0 & 0 & 0 & 0 \\ \hline
    \end{tabular}
    \normalsize\\
{\Large $0^{\frac{n}{2}}1^{\frac{n}{8}}0^{\frac{3n}{8}}$} & {\Large $\{0\}^{\frac{n}{2}}\{0, 1\}^{\frac{n}{4}}\{0\}^{\frac{3n}{8}}$} & 
    \small
    \tabcolsep=0.12cm
    \begin{tabular}{|l|l|l|l|l|l|l|l|l|l|l|l|l|l|l|l|}
    \hline
    0 & 0 & 0 & 0 & 0 & 0 & 0 & 0 &\cellcolor{blue_color}$\star$ &\cellcolor{blue_color}$\star$ & 0 & 0 & 0 & 0 & 0 & 0 \\ \hline
    \end{tabular}
    \normalsize\\
{\Large $0^{\frac{3n}{4}}1^{\frac{n}{8}}0^{\frac{n}{8}}$} & {\Large $\{0\}^{\frac{3n}{4}}\{0, 1\}^{\frac{n}{8}}\{0\}^{\frac{n}{8}}$} & 
    \small
    \tabcolsep=0.12cm
    \begin{tabular}{|l|l|l|l|l|l|l|l|l|l|l|l|l|l|l|l|}
    \hline
    0 & 0 & 0 & 0 & 0 & 0 & 0 & 0 & 0 & 0 & 0 & 0 &\cellcolor{blue_color}$\star$ & \cellcolor{blue_color}$\star$ & 0 & 0 \\ \hline
    \end{tabular}
    \normalsize
\end{tabular}
\caption{Structure of $\hat\R$. Blue cells with $\star$ represent indices in which either a 0 or a 1 could appear.}
\label{fig:struct-os}
\end{center}
\end{figure}

This $\hat{\R}$ describes all the strings used as the source of randomness in the $O(\log{n} \log{\log{n}})$ communication protocol for $\GT$, but each of the strings appears in the structure only once instead of $\alpha$ times. So, we need to duplicate this structure $\alpha$ times to properly simulate the protocol.

To finish the structure, we are going to add two other steps to the structure. 
First, we are going to have some number $t$ of individual ``prepackaged'' copies to be determined later for the $\GT$ protocol. Let $\R_1 = \ldots = \R_{t} = \times_{\alpha} \hat{\R}$. Each of the copies has enough randomness and the right structure of that randomness to simulate one full run of the $\GT$ protocol. Let $\R' = \bigtimes_{j \in [t]} \R_{j}$ which allows us to handle $t$ runs.
Secondly, we want to be able to obtain the value of any specific index of $x$, so we add a set of ``basis'' strings to the structure: $\mathcal{B} = \{\one_1\} \times \{\one_2\} \times \ldots \times \{\one_n\}= \{10\ldots0\}\times\{010\ldots0\} \times \ldots \times \{00\ldots010\} \times \{00\ldots01\}$. 

The final underlying structure of the oracle will be a Cartesian product of $\R'$ and $\mathcal{B}$: 
$\R_{\OS} = \mathcal{B} \times \R' = \mathcal{B} \times (\bigtimes_{j \in [t]} \R_{j})$. See Figure \ref{fig:oracle-os} for an illustration.

\begin{figure}[h]
\begin{tabular}{ccccccccccc}
$\mathcal{B}$ &             & $\R_1$ &             & $\R_2$ &             &          &             & $\R_{t-1}$ &             & $\R_t$ \\
 &             & $\alpha$ copies &             & $\alpha$ copies &             &          &             & $\alpha$ copies &             & $\alpha$ copies \\
\footnotesize
\tabcolsep=0.11cm
\begin{tabular}{|llllll|}
\hline
$\cellcolor{yellow}1$ & 0 & 0 & 0 & 0 & 0 \\
0 & $\cellcolor{yellow}1$ & 0 & 0 & 0 & 0 \\
0 & 0 & \cellcolor{yellow}1 & 0 & 0 & 0 \\
0 & 0 & 0 & \cellcolor{yellow}1 & 0 & 0 \\
0 & 0 & 0 & 0 & \cellcolor{yellow}1 & 0 \\
0 & 0 & 0 & 0 & 0 & \cellcolor{yellow}1 \\ \hline
\end{tabular}
& $\bigtimes$ &
\footnotesize
\tabcolsep=0.11cm
\begin{tabular}{|llllllll|}
\hline
$\cellcolor{blue_color}$ & $\cellcolor{blue_color}$ & $\cellcolor{blue_color}$ & $\cellcolor{blue_color}$ & & & & \\
$\cellcolor{blue_color}$ & $\cellcolor{blue_color}$ & & & & & & \\
& & & & $\cellcolor{blue_color}$ & $\cellcolor{blue_color}$ & & \\
$\cellcolor{blue_color}$ & & & & & & & \\
& & $\cellcolor{blue_color}$ & & & & & \\
& & & & $\cellcolor{blue_color}$ & & & \\
& & & & & & $\cellcolor{blue_color}$ & \\ \hline
\end{tabular}
& $\bigtimes$ & 
\footnotesize
\tabcolsep=0.11cm
\begin{tabular}{|llllllll|}
\hline
$\cellcolor{blue_color}$ & $\cellcolor{blue_color}$ & $\cellcolor{blue_color}$ & $\cellcolor{blue_color}$ & & & & \\
$\cellcolor{blue_color}$ & $\cellcolor{blue_color}$ & & & & & & \\
& & & & $\cellcolor{blue_color}$ & $\cellcolor{blue_color}$ & & \\
$\cellcolor{blue_color}$ & & & & & & & \\
& & $\cellcolor{blue_color}$ & & & & & \\
& & & & $\cellcolor{blue_color}$ & & & \\
& & & & & & $\cellcolor{blue_color}$ & \\ \hline
\end{tabular}
& $\bigtimes$ & $\ldots$ & $\bigtimes$ & 
\footnotesize
\tabcolsep=0.11cm
\begin{tabular}{|llllllll|}
\hline
$\cellcolor{blue_color}$ & $\cellcolor{blue_color}$ & $\cellcolor{blue_color}$ & $\cellcolor{blue_color}$ & & & & \\
$\cellcolor{blue_color}$ & $\cellcolor{blue_color}$ & & & & & & \\
& & & & $\cellcolor{blue_color}$ & $\cellcolor{blue_color}$ & & \\
$\cellcolor{blue_color}$ & & & & & & & \\
& & $\cellcolor{blue_color}$ & & & & & \\
& & & & $\cellcolor{blue_color}$ & & & \\
& & & & & & $\cellcolor{blue_color}$ & \\ \hline
\end{tabular}
& $\bigtimes$ & 
\footnotesize
\tabcolsep=0.11cm
\begin{tabular}{|llllllll|}
\hline
$\cellcolor{blue_color}$ & $\cellcolor{blue_color}$ & $\cellcolor{blue_color}$ & $\cellcolor{blue_color}$ & & & & \\
$\cellcolor{blue_color}$ & $\cellcolor{blue_color}$ & & & & & & \\
& & & & $\cellcolor{blue_color}$ & $\cellcolor{blue_color}$ & & \\
$\cellcolor{blue_color}$ & & & & & & & \\
& & $\cellcolor{blue_color}$ & & & & & \\
& & & & $\cellcolor{blue_color}$ & & & \\
& & & & & & $\cellcolor{blue_color}$ & \\ \hline
\end{tabular}
\end{tabular}
\caption{Structure of $\R_{\OS}$. Each $\R_j$ consist of $\alpha$ copies of $\hat{\R}$. }
\label{fig:oracle-os}
\end{figure}

\normalsize

We also set the parameters to be $\alpha = 2\log{(\log{n})}, t = 250n\ln{2}$.
Notice that this set $\R_{\OS}$ is a \emph{good base} by construction and has size $m = n + \alpha t n = n + c n^2 \log{(\log{n})}$ for some constant $c$.

\textbf{Constructing the family of approximating polynomials.} 
In order to prove this claim, we first describe a randomized query algorithm that computes $\GT_i(x)$ correctly for all $i$ and $x$ with high probability given $\OracleWithParameter{\R_{\OS}}(r, x)$ as input. We then explain how to convert this query algorithm into a polynomial. 
The algorithm construction itself consists of two parts. In the first part, for all $j \in [t]$ we show the existence of a deterministic algorithm $B_{(r, i, j)}$ that, given $\Oracle(r, x)$, can compute $\GT_i(x)$ for every specific $x, i \in \{0, 1\}^n$ with good probability over the choice of $r \from \R_{\OS}$, and this algorithm is only going to use the parts of the input that correspond to $\R_j$ and $\mathcal{B}$. In the second part, we show that the algorithm $A_{(r, i)}$ that chooses a copy $j$ to use randomly and runs $B_{(r, i, j)}$, computes $\GT_i(x)$ correctly for all $i$ and $x$ with high probability given $\Oracle(r, x)$ as input.
        
    For all $i \in \{0, 1\}^n, j \in [t], r \in \R_{\OS}$ let $B_{(r, i, j)}(\Oracle(r, x))$ be the following deterministic algorithm. 

    \begin{enumerate}
        \item Set $\ell = 0, u = n/2$
        \item\label{step:os-ub} While $\ell < u$:
        \item \qquad Set $\tau = 0^{\ell} 1^{u - \ell} 0^{n - u}$
        \item \qquad For all indices $v \in [m]$ corresponding to $n$-bit strings drawn from $\R^{\tau}$ within the $j$-th 
        
        \qquad copy $\R_j$:
        \item \qquad \qquad Compute $\langle i, r_v \rangle$ and compare it to $(\Oracle(r, x))_v = \langle x, r_v \rangle$. 
        \item  \qquad If for all such $v$ the inner products are equal, i.e., $\langle i, r_v \rangle = (\Oracle(r, x))_v$, then set 
        
        \qquad $\textit{tmp} = u, u = u + (u - \ell)/2, \ell = \textit{tmp}$ and go step~\ref{step:os-ub}.
        \item  \qquad Otherwise, set $u = (u + \ell)/2$ and go step~\ref{step:os-ub}
        \item Compare $i_{\ell} = \langle i, \one_{\ell} \rangle$ and $(\Oracle(r, x))_{\ell} = \langle x, \one_{\ell} \rangle = x_{\ell}$. If $x_{\ell} \leq i_{\ell}$ then accept. Otherwise, reject.
    \end{enumerate}

    The last step is possible specifically because of $\mathcal{B}$ in the structure of $\R_{\OS}$: $r_{\ell} = \one_{\ell}$ for all $\ell \leq n$ and for all $r \in \R_{\OS}$.
    Notice that this algorithm emulates the randomized communication protocol for the $\GT$ communication problem. 
    
    In general, the algorithm emulates the randomized communication protocol for equality on the first half of the segment $[\ell, u + (u - \ell)]$ in $x$ and $i$, and depending on the result it splits the inputs into smaller segments and continues recursively. In the end, if all the runs of equality protocols were correct, the algorithm finds and compares the most significant bit where $x$ and $i$ differ.

    By \cite{nisan1993communication} we know that this algorithm computes $\GT_i(x)$ with probability at least $1 - (\log{n})2^{-\alpha} = 1 - (\log{n})2^{-2\log{(\log{n}})} = 1 - \frac{1}{\log{n}} \geq \frac{11}{12}$ for sufficiently large $n$ independently of the choice of $j \in [t]$. That is, for all $j \in [t]$ and for all $i, x \in \{0, 1\}^n$,

    $$\Pr_{r \from \R_{\OS}}[B_{(r, i, j)}(\Oracle(r, x)) = \GT_i(x)] \geq \frac{11}{12}.$$
    This algorithm makes at most $\alpha \log{n} = 2 \log{n} \log{\log{n}}$ queries to the oracle $\Oracle(r, x)$. Note that this algorithm needs access to the specific $r$ needed to compute every $\langle i, r_v\rangle$ and we enable this by ``hardcoding'' this $r$ into the algorithm and creating a separate algorithm for each possible $r$.

    We have shown that for every fixed $i, x \in \{0, 1\}^n$ there are many $r \in \R_{\OS}$ that if used as a first input for the oracle $\Oracle$ allow $B_{(r, i, j)}$ to compute $\GT_i(x)$. Unfortunately, this is not enough: our algorithm should be universal, i.e., we want a single algorithm that with high probability over $r$ succeeds on all $i$ and $x$. On the other hand, $B_{(r, i, j)}$ only uses one fixed ``package'' of random strings, namely the $j$-th package. 

    Let $W(i, x, r, j)$ be the indicator that the $j$-th package of random strings in $r$ defines a set of ``bad'' random strings for $(i, x)$: $W(i, x, r, j) = 1$ if and only if $B_{(r, i, j)}(\Oracle(r, x)) \neq \GT_i(x)$. We established that $B_{(r, i, j)}(\Oracle(r, x))$ works well if given a random $r \from \R_{\OS}$ for every $j \in [t]$ and the probability of this algorithm outputting an incorrect answer is at most $\frac{1}{12}$. So for all $i, x \in \{0, 1\}^n, j \in [t]$, we have
    $$\Pr_{r \from \R_{\OS}}[W(i, x, r, j) = 1] = \E_{r \from \R_{\OS}}[W(i, x, r, j)] \leq \frac{1}{12}.$$
    We can't immediately get a useful upper bound on the probability of $r \from \R$ working out for all $i$ and $x$ at the same time. To achieve this, we'll design a new algorithm that uses all $t$ packages of random strings. Its construction and analysis are inspired by Newman's classic argument used for simulating public randomness by private randomness in communication protocols.
    
    For all $i \in \{0, 1\}^n, r \in \R_{\OS}$ let $A_{(r, i)}(\Oracle(r, x))$ be the following randomized algorithm:
        
    \begin{itemize}
        \item Choose $j \from [t]$ uniformly at random.
        \item Run $B_{(r, i, j)}(\Oracle(r, x))$.
    \end{itemize}
    Let us now analyse $A_{(r, i)}$. The number of queries that $A_{(r, i)}$ makes to the oracle is the same as $B_{(r, i, j)}$ which is $\alpha \log{n} = 2 \log{n} \log{\log{n}}$. We fix a pair $(i, x)$ and evaluate the following probability.
    $$ \Pr_{r \from \R_{\OS}}\left[\Pr_{j \from [t]}[B_{(r, i, j)} \neq \GT_i(x)] > \frac{1}{6}\right] = \Pr_{r \from \R_{\OS}}\left[\frac{1}{t} \sum_{j \in [t]} W(i, x, r, j) > \frac{1}{6}\right].$$
    We established that $\E_{r \from \R_{\OS}}[W(i, x, r, j)] \leq \frac{1}{12}$ and so by Hoeffding's inequality,
    $$\Pr_{r \from \R_{\OS}}\left[\frac{1}{t} \sum_{j \in [t]} W(i, x, r, j) > \frac{1}{12} + \frac{1}{12}\right] \leq e^{-2\frac{t}{144}} \leq 2^{-\frac{500n}{144}}.$$
    By a union bound over all possible $i, x \in \{0, 1\}^n$,
    $$\Pr_{r \from \R_{\OS}}\left[\exists i, x \in \{0, 1\}^n: \frac{1}{t} \sum_{j \in [t]} W(i, x, r, j) > \frac{1}{6}\right] \leq 2^{2n} 2^{-\frac{500n}{144}} \leq 2^{-n} < \frac{1}{3}.$$
    Therefore, we have proven that
    $$\Pr_{r \from \R_{\OS}}\left[\exists i, x \in \{0, 1\}^n: \Pr_{j \from [t]}[A_{(r, i)}(\Oracle(r, x))) \neq \GT_i(x)] > \frac{1}{6}\right] < \frac{1}{3}.$$
    The last step is to convert this family of query algorithms into a family of approximating polynomials. Let $q_{(r, i)}$ denote the acceptance probability of $A_{(r, i)}$. A standard argument (e.g.,~\cite[Theorem 15]{BuhrmanW02}) implies that this is a polynomial of degree at most $2 \log{n} \log{\log{n}}$ such that
        
    $$\Pr_{r \from \R_{\OS}}\left[\exists i, x \in \{0, 1\}^n: \left|q_{(r, i)}(\Oracle(r, x)) - \GT_i(x)\right| > \frac{1}{6}\right] < \frac{1}{3},$$
    which is exactly what we were looking for.
\end{proof}

We successfully converted the most well-known communication protocol for $\GT$ that requires $O(\log{n} \log{\log{n}})$ bits of communication into a family of polynomials of degree $O(\log{n} \log{\log{n}})$ that approximates $\GT_i$. It's known that there is a better communication protocol for $\GT$ that requires only $O(\log{n})$ bits of communication, as observed by Nisan \cite{nisan1993communication}. The next claim establishes that this more efficient protocol can be converted into a family of polynomials as well.

\begin{claim}\label{clm:os2-ub}
   There exists a \emph{good base} $\R_{\OS++}$ of size $m = O(n^3 \log{n})$ such that if $r \from \R_{\OS++}$ is sampled uniformly at random, then with probability at least $\frac{2}{3}$ over the choice of $r$ there exists a family of polynomials $\{q_{(r, i)}: \{0, 1\}^{m} \to \{0, 1\} \mid i \in \{0, 1\}^n\}$, each of degree at most $O(\log{n})$, such that given $\Oracle(r, x)$ as the input, each polynomial $q_{(r, i)}(\OracleWithParameter{\R_{\OS++}}(r, x))$ approximates the corresponding $\GT_i(x)$ with error at most $\frac{1}{6}$. That is,
   
    $$\Pr_{r \from \R_{\OS++}}\left[\exists i, x \in \{0, 1\}^n: \left|q_{(r, i)}(\Oracle(r, x)) - \GT_i(x)\right| > \frac{1}{6}\right] < \frac{1}{3}.$$  
\end{claim}
The proof of Claim \ref{clm:os2-ub} is similar to the proof of Claim \ref{clm:os-ub} and can be found in Appendix \ref{sec:os2-ub}.

\subsection{General lower bound}

To complete the framework and to obtain the lower bound for Ordered Search we need to show why computing the parity is hard even given $\OracleWithParameter{\R_{\OS}}$ or $\OracleWithParameter{\R_{\OS++}}$. We will show a stronger lower bound that would allow us to reuse this lower bound for other applications. Specifically, we will show that computing the parity of input $x$ remains hard given $\OracleWithParameter{\R}$ for any \emph{good base} $\R$ of small size.

\subsubsection{Combinatorial claim}

    The hardness of $\parity$ in this model is based on the following statement.
    For every \emph{good base} $\R$ of small size with high probability over the sample $r \from \R$ for every set of $n$-bit strings taken from the collection $r$ of size at most $O(\frac{n}{\log{n}})$, the bitwise parity of these strings is not equal to the all-ones string.
    
\begin{claim}\label{clm:gen-comb}
    For every \emph{good base} $\R$ of size $m$ with probability at least $\frac{2}{3}$ over the choice of $r \from \R$ for every set of elements $T \subseteq [m]$ of size at most $d = \frac{n}{4\log{m}} - 1$, the bitwise parity of $n$-bit strings $r_i$, $i \in T$ from the collection $r \from \R$ is not equal to the all-ones string:
    
    $$\Pr_{r \from \R}\left[\forall T \subseteq [m], |T| \leq d: \bigoplus_{i \in T} r_i \neq 1^n\right] \geq \frac{2}{3}.$$
    
\end{claim}
\begin{proof}
    Fix an arbitrary \emph{good base} $\R$ of size $m$.
    Fix a set $T \subseteq [m]$ where $|T| \leq d$. We want to bound the probability $\Pr_{r \from \R}[\bigoplus_{i \in T} r_i = 1^n]$ that for this $r$ and for this $T$ the strings corresponding to the indices in $T$ sum up to the string of all ones. 
    Fix a specific index $k \in [n]$. We compute the probability that index $k$ is set to 1 in $\bigoplus_{i \in T} r_i$. To do this we need to understand how the candidate strings $r_i, i\in T$ can influence this value. 
    
    There are three possible scenarios for each index $k$: 
        \begin{itemize}
        \item (Type I) There is at least one string $r_i \in \{0, 1\}^n$ with $i \in T$ such that it is chosen from $\R^{\tau}$ where $\tau_k = 1$. Then in each such string, the bit at index $k$ is sampled independently at random with probability $\frac{1}{2}$. Thus $\Pr_{r \from \R}[\langle \bigoplus_{i \in T} r_i, \one_k\rangle = 1] = \frac{1}{2}$. 
        \item (Type II) There are no strings $r_{i}, i \in T$ such that $r_i$ is chosen from $\R^{\tau}$ and $\tau_k = 1$, but there is $r_i, i \in T$ that is chosen from $\mathcal{B}$, such that $r_i = \one_k$. Then the value of $\langle \bigoplus_{i \in T} r_i, \one_k\rangle$ is one since there is exactly one string in this sum with the $k$th index value set to one. Thus $\Pr_{r \from \R}[\langle \bigoplus_{i \in T} r_i, \one_k\rangle = 1] = 1$.
        \item (Type III) There are no strings $r_{i}, i \in T$ such that $r_i$ is chosen from $\R^{\tau}$ and $\tau_k = 1$, and there is no $r_i, i \in T$ that is chosen from $\mathcal{B}$, such that $r_i = \one_k$. Then for all strings $r_{i}$ the index $k$ is 0, so $\Pr_{r \from \R}[\langle \bigoplus_{i \in T} r_i, \one_k\rangle = 1] = 0$.
        \end{itemize}    
\begin{center}
\begin{figure}[h]
\begin{tabular}{ccccccccc}
\textbf{{\large Index $k$}}                       & $1$                                                          & $2$                                                          & $3$                                                          & $4$                                                          & $5$                                                          & $6$                                                          & $7$                                                          & $8$                                                          \\ \cline{2-9} 
\multicolumn{1}{r|}{$r_{i_1}$}                                     & \multicolumn{1}{c|}{$\star\cellcolor{blue_color}$}                     & \multicolumn{1}{c|}{$\star\cellcolor{blue_color}$}                     & \multicolumn{1}{c|}{$\star\cellcolor{blue_color}$}                     & \multicolumn{1}{c|}{$\star\cellcolor{blue_color}$}                     & \multicolumn{1}{c|}{$0$}                                     & \multicolumn{1}{c|}{$0$}                                     & \multicolumn{1}{c|}{$\star\cellcolor{blue_color}$}                                     & \multicolumn{1}{c|}{$0$}                                     \\ \cline{2-9} 
                                                                   &                                                              &                                                              &                                                              &                                                              &                                                              &                                                              &                                                              &                                                              \\ \cline{2-9} 
\multicolumn{1}{r|}{$r_{i_2}$}                                     & \multicolumn{1}{c|}{$\star\cellcolor{blue_color}$}                     & \multicolumn{1}{c|}{$\star\cellcolor{blue_color}$}                     & \multicolumn{1}{c|}{$\star\cellcolor{blue_color}$}                                     & \multicolumn{1}{c|}{$0$}                                     & \multicolumn{1}{c|}{$0$}                                     & \multicolumn{1}{c|}{$0$}                                     & \multicolumn{1}{c|}{$0$}                                     & \multicolumn{1}{c|}{$0$}                                     \\ \cline{2-9} 
                                                                   &                                                              &                                                              &                                                              &                                                              &                                                              &                                                              &                                                              &                                                              \\ \cline{2-9} 
\multicolumn{1}{r|}{$r_{i_3}$}                                     & \multicolumn{1}{c|}{$\star\cellcolor{blue_color}$}                                     & \multicolumn{1}{c|}{$0$}                                     & \multicolumn{1}{c|}{$\star\cellcolor{blue_color}$}                     & \multicolumn{1}{c|}{$\star\cellcolor{blue_color}$}                     & \multicolumn{1}{c|}{$0$}                                     & \multicolumn{1}{c|}{$0$}                                     & \multicolumn{1}{c|}{$\star\cellcolor{blue_color}$}                                     & \multicolumn{1}{c|}{$0$}                                     \\ \cline{2-9} 
                                                                   &                                                              &                                                              &                                                              &                                                              &                                                              &                                                              &                                                              &                                                              \\ \cline{2-9} 
\multicolumn{1}{r|}{$r_{i_4}$}                                     & \multicolumn{1}{c|}{$0$}                     & \multicolumn{1}{c|}{$\star\cellcolor{blue_color}$}                                     & \multicolumn{1}{c|}{$0$}                                     & \multicolumn{1}{c|}{$0$}                                     & \multicolumn{1}{c|}{$0$}                                     & \multicolumn{1}{c|}{$0$}                                     & \multicolumn{1}{c|}{$0$}                                     & \multicolumn{1}{c|}{$0$}                                     \\ \cline{2-9} 
                                                                   &                                                              &                                                              &                                                              &                                                              &                                                              &                                                              &                                                              &                                                              \\ \cline{2-9} 
\multicolumn{1}{r|}{$r_{i_5}$}                                     & \multicolumn{1}{c|}{$0$}                                     & \multicolumn{1}{c|}{$0$}                                     & \multicolumn{1}{c|}{$1\cellcolor{yellow}$}                                     & \multicolumn{1}{c|}{$0$}                                     & \multicolumn{1}{c|}{$0$}                                     & \multicolumn{1}{c|}{$0$}                                     & \multicolumn{1}{c|}{$0$}                                     & \multicolumn{1}{c|}{$0$}                                     \\ \cline{2-9} 
                                                                   &                                                              &                                                              &                                                              &                                                              &                                                              &                                                              &                                                              &                                                              \\ \cline{2-9} 
\multicolumn{1}{r|}{$r_{i_6}$}                                     & \multicolumn{1}{c|}{$0$}                                     & \multicolumn{1}{c|}{$0$}                                     & \multicolumn{1}{c|}{$0$}                                     & \multicolumn{1}{c|}{$0$}                                     & \multicolumn{1}{c|}{$1\cellcolor{yellow}$}                                     & \multicolumn{1}{c|}{$0$}                                     & \multicolumn{1}{c|}{$0$}                                     & \multicolumn{1}{c|}{$0$}                                     \\ \cline{2-9} 
\multicolumn{1}{l}{}                                               & \multicolumn{1}{l}{}                                         & \multicolumn{1}{l}{}                                         & \multicolumn{1}{l}{}                                         & \multicolumn{1}{l}{}                                         & \multicolumn{1}{l}{}                                         & \multicolumn{1}{l}{}                                         & \multicolumn{1}{l}{}                                         & \multicolumn{1}{l}{}                                         \\
\textbf{{\large Type}}                            & I                                                            & I                                                            & I                                                            & I                                                            & II                                                          & III                                                          & I                                                           & III                                                          \\
\multicolumn{1}{l}{}                                               & \multicolumn{1}{l}{{\color{white} $1/2$}} & \multicolumn{1}{l}{{\color{white} $1/2$}} & \multicolumn{1}{l}{{\color{white} $1/2$}} & \multicolumn{1}{l}{{\color{white} $1/2$}} & \multicolumn{1}{l}{{\color{white} $1/2$}} & \multicolumn{1}{l}{{\color{white} $1/2$}} & \multicolumn{1}{l}{{\color{white} $1/2$}} & \multicolumn{1}{l}{{\color{white} $1/2$}} \\
$\Pr_{r \from \R}[\langle \bigoplus_{i \in T} r_i, \one_k\rangle]$ & $1/2$                                                        & $1/2$                                                        & $1/2$                                                        & $1/2$                                                        & $1$                                                          & $0$                                                          & $1/2$                                                          & $0$                                                         
\end{tabular}
\caption{Example of index types, $T = \{i_1, i_2, i_3, i_4, i_5, i_6\}$.}
\label{fig:types}
\end{figure}    
\end{center}       
Notice that $T$ fully defines the types of all indices and thus the values of $\langle \bigoplus_{i \in T} r_i, \one_k\rangle$ for $k$ of types II and III don't depend on the choice of $r \from \R$. On the other hand, the values of indices of type I do depend on the choice of $r \from \R$. Each of them is either a parity of independent random bits or the negation of a parity of independent random bits which is fixed by $T$ too. Thus they behave as independent bits themselves and therefore the values $\langle \bigoplus_{i \in T} r_i, \one_k\rangle$ are mutually independent for all indices $k$.

Denote by $n_{\text{I}}, n_{\text{II}}, n_{\text{III}}$ the numbers of indices of each type. Notice that $n_{\text{I}}+n_{\text{II}}+n_{\text{III}} = n$ and  $n_{\text{II}} \leq d$. Then in this notation

        $$\Pr_{r \from \R}\left[\bigoplus_{i \in T} r_i = 1^n\right] = \left(\frac{1}{2}\right)^{n_{\text{I}}} 1^{n_{\text{II}}} 0^{n_{\text{III}}}.  $$

 If there exists $k \in [n]$ of the third type, the probability $\Pr_{r \from \R}[\bigoplus_{i \in T} r_i = 1^n]$ becomes $0$, so to upper bound the probability we may assume all the indices have one of the first two types. And, since $n_{\text{II}} \leq d$, to maximize the value we assume that $n_{\text{II}} = d$. Thus we have
        
        $$\Pr_{r \from \R}\left[\bigoplus_{i \in T} r_i = 1^n\right] = \left(\frac{1}{2}\right)^{n - n_{\text{II}}} 1^{n_{\text{II}}} \leq 1^d \left(\frac{1}{2}\right)^{n-d} = 2^{-(n-d)}.$$
        Since $d = \frac{n}{4\log{m}}-1$ and $m \geq n$, we have $n-d > \frac{n}{2}$ for sufficiently large $n$. So for a fixed $T$,
        
        $$\Pr_{r \from \R}\left[\bigoplus_{i \in T} r_i = 1^n\right] < 2^{-\frac{n}{2}}.$$
    There are $\binom{m}{\leq d}$ ways to choose the set $T$, so by a union bound over the choice of $T$, the probability that for some set of size at most $d$ the value $\bigoplus_{i \in T} r_i$ is equal to the string of all ones is
    
    $$\Pr_{r \from \R}\left[\exists T \subseteq [m], |T| \leq d: \bigoplus_{i \in T} r_i = 1^n\right] \leq \binom{m}{\leq d} 2^{-\frac{n}{2}} = 2^{-\frac{n}{2}} \sum_{d' = 0}^d \binom{m}{d'} \leq 2^{-\frac{n}{2}} \sum_{d' = 0}^d m^{d'} \leq 2^{-\frac{n}{2}} m^{d+1} $$ $$ = 2^{-\frac{n}{2}} m^{\frac{n}{4\log{m}}} = 2^{-\frac{n}{2}} 2^{\frac{n}{4\log{m}} \log{m}} = 2^{\frac{n}{4} - \frac{n}{2}} = 2^{-\frac{n}{4}} < \frac{1}{3}.$$
\end{proof}

\subsubsection{Lower bound on the degree of $\CRAPWithParameter{\R}$ }

For every \emph{good base} $\R$ and for every fixed $r \in \R$ define the function $\CRAPWithParameter{\R}_{r}: D[\R]_r \to \{0, 1\}$ where $D[\R]_r = \{\OracleWithParameter{\R}(r, x) \mid x \in \{0, 1\}^n\}$ is the subset of $\{0, 1\}^m$ where each domain point corresponds to one specific $x \in \{0, 1\}^n$ and is consistent with the fixed $r$. This function outputs the parity of the string encoded by the input: $\CRAPWithParameter{\R}_{r}(\Oracle(r, x)) = \parity(x)$. It is well defined since $\parity(x) = \bigoplus_{r_i \in \mathcal{B}} \langle x, r_i\rangle = \bigoplus_{i \in [n]} (\Oracle(r, x))_i$. Note that both $D[\R]_r$ and $\CRAPWithParameter{\R}$ are parameterized by $\R$ and, as with $\OracleWithParameter{\R}$, we will omit the parameter later in places where the parameter is clear from the context.

Our goal is to show that $\CRAPWithParameter{\R}$ is hard to approximate if $\R$ is a \emph{good base} of small size. We do this by showing that for every \emph{good base} $\R$ of size $m$ if $r \from \R$ u.a.r. then every polynomial $p$ of degree at most $d = O(\frac{n}{\log{m}})$ is completely uncorrelated with $\CRAPWithParameter{\R}_{r}(\Oracle(r, x))$ with high probability over the choice of $r$. 

\begin{theorem}\label{thm:gen-lb}
    For every \emph{good base} $\R$ of size $m$ if $r \from \R$ u.a.r. then with probability at least $\frac{2}{3}$ over the choice of $r$ every polynomial $p : \{0, 1\}^{m} \to \reals$ of degree at most $d = \frac{n}{4\log{m}} - 1$ doesn't approximate $\CRAPWithParameter{\R}_r$:

    $$\Pr_{r \from \R}\left[\forall \varepsilon < \frac{1}{2}, \ \forall p, \deg(p) \leq d, \exists \oracle \in D[\R]_r: \left|p(\oracle) - \CRAPWithParameter{\R}_r(\oracle)\right| > \varepsilon\right] \geq \frac{2}{3}$$
    
\end{theorem}

Note that Theorem~\ref{thm:gen-lb} rules out approximating polynomials that may be unbounded outside of the domain of $\CRAPWithParameter{\R}_r$. That is, it asserts that there is no low-degree approximating polynomial even when that polynomial is permitted to take values outside of $[0, 1]$ on points outside of the domain of $\CRAP_{r}$. Note also that since the lower bound applies for all $\eps < 1/2$, it actually entails a threshold degree lower bound on computing $\CRAPWithParameter{\R}$.


\begin{proof} 

    Fix an arbitrary \emph{good base} $\R$ of size $m$.

    For convenience in this proof, let us change notation to consider polynomials approximations over $\{-1, 1\}$ instead of over $\{0, 1\}$. Define $\Oracle': \R \times \{-1, 1\}^n \to \{-1, 1\}^{m}$ to be $(\Oracle'(r, x'))_i = 1 - 2(\Oracle(r, (\frac{1-x'_1}{2}, \frac{1-x'_2}{2}, \ldots, \frac{1-x'_n}{2})))_i = 1-2\langle x, r_i\rangle$ where $r_i$ is the vector corresponding to $i$th component of $\Oracle(r, x)$ and $x \in \{0, 1\}^n$ is the vector that corresponds to $x' \in \{-1, 1\}^n$: $x_i = \frac{1-x'_i}{2}$ for all $i \in [n]$. Notice that this change of notation satisfies the following: if $a \in \{0, 1\}$ and $a'$ is the corresponding value in the new notation $a' \in \{-1, 1\}$ then $a' = (-1)^{a}$.

    Let's also rewrite $\CRAP_{r}$ in this new notation. Let $D'_r$ represent the domain of $\CRAP'_r$: $D'_r = \{\Oracle'(r, x') \mid x' \in \{-1, 1\}^n\}$ and the function $\CRAP'_{r}: D'_r \to \{-1, 1\}$ be $\CRAP'_{r}(\Oracle'_1, \Oracle'_2, \ldots, \Oracle'_m) = 1 - 2 \CRAP_{r}(\frac{1-\Oracle'_1}{2}, \frac{1-\Oracle'_2}{2}, \ldots, \frac{1-\Oracle'_m}{2})$.

    Note that every polynomial $p': \{-1, 1\}^m \to \reals$ that approximates $\CRAP'_r$ to error $\varepsilon$ can be converted by a linear transformation into a polynomial $p: \{0, 1\}^m \to \reals$ of the same degree that approximates $\CRAP_r$ to error $\varepsilon/2$. So it suffices to prove that no polynomial $p'$ of degree at most $d$ approximates $\CRAP'_r$ to error $\varepsilon < 1$.
    
    Assume toward a contradiction that there is a polynomial $p'$ of degree $d$ that approximates $\CRAP'_{r}$. This means that there exists $\varepsilon < 1$ such that for all $\oracle' \in D'_r$,
    
    $$\left| p'(\oracle') - \CRAP'_{r}(\oracle')\right| < \varepsilon.$$
    
    Consider the following expression:

\begin{equation} \label{eq:1}
\begin{split}
\frac{1}{2^n} \left|\sum_{\oracle' \in D'_r} \CRAP'_{r}(\oracle') (\CRAP'_{r}(\oracle') - p'(\oracle'))\right| &\leq \frac{1}{2^n} \left(\max_{\oracle' \in D'_r} \left|p'(\oracle') - \CRAP'_{r}(\oracle')\right|\right) \left(\sum_{\oracle' \in D'_r} \left|\CRAP'_{r}(\oracle')\right| \right)\\ & < \frac{1}{2^n}\varepsilon |D'_r| = \varepsilon.
\end{split}
\end{equation}
    The last equality holds because $\Oracle'(r, \cdot)$ is surjective, and hence $|D'_r| = 2^n$. 
    On the other hand,

\begin{equation} 
\begin{split}
    \frac{1}{2^n} \left|\sum_{\oracle' \in D'_r} \CRAP'_{r}(\oracle') (\CRAP'_{r}(\oracle') - p'(\oracle'))\right| &= 
    \frac{1}{2^n} \left|\left(\sum_{\oracle' \in D'_r}\CRAP'_{r}(\oracle') \CRAP'_{r}(\oracle')\right) - \left(\sum_{\oracle' \in D'_r} \CRAP'_{r}(\oracle') p'(\oracle')\right)\right| \label{eq:2}\\ &= \frac{1}{2^n}\left||D'_r| - \left(\sum_{\oracle' \in D'_r} \CRAP'_{r}(\oracle') p(\oracle')\right)\right|.
\end{split}
\end{equation}

    We now show that with high probability the expression above is equal to $\frac{|D'_r|}{2^n}$.

    \begin{claim}\label{clm:gen-lb-helper}
        With probability at least $\frac{2}{3}$ over the choice of $r \from \R$, for every polynomial $p' : \{-1, 1\}^m \to \reals$ of degree at most $d = \frac{n}{4\log{m}} - 1$ we have $$\sum_{\oracle' \in D'_r} \CRAP'_{r}(\oracle') p'(\oracle') = 0.$$
    \end{claim}
    \begin{proof}

    Fix a polynomial $p'$ of degree at most $d = \frac{n}{4\log{m}} - 1$.
    By linearity it suffices to consider the case where $p'$ is a monomial, $p'(\oracle') = \prod_{j \in T} \oracle'_{j}$ for some $T \subseteq [m], |T| \leq d$. So 
    
    $$\sum_{\oracle' \in D'_r} \CRAP'_{r}(\oracle') p'(\oracle') = \sum_{x' \in \{-1, 1\}^n} \left(\prod_{i \in [n]} (x'_i) \right)\left(\prod_{j \in T} (\Oracle'(r, x'))_j\right) = \sum_{x \in \{0, 1\}^n}\left((-1)^{\langle x, 1^n\rangle} \right) \left(\prod_{j \in T} (-1)^{\langle x, r_{j}\rangle}\right) $$ $$ = \sum_{x \in \{0, 1\}^n}(-1)^{\langle x, 1^n\rangle}  (-1)^{\sum_{j \in T} \langle x, r_{j}\rangle} = \sum_{x \in \{0, 1\}^n}(-1)^{\langle x, 1^n\rangle}  (-1)^{\langle x, \bigoplus_{j \in T} r_{j}\rangle} = \sum_{x \in \{0, 1\}^n} (-1)^{\langle x, 1^n \oplus (\bigoplus_{j \in T} r_j)\rangle}$$
    
    This expression is not zero if and only if $\bigoplus_{j \in T} r_{j} = 1^n$. By Claim~\ref{clm:gen-comb} the probability that such $T$ exists is at most $\frac{1}{3}$. So the probability over the choice of $r$ for some polynomial $p' : \{-1, 1\}^m \to \reals$ of degree at most $d = \frac{n}{4\log{m}} - 1$ to have $$\sum_{\oracle' \in D'_r} \CRAP'_{r}(\oracle') p'(\oracle') \neq 0$$ 
    is at most $\frac{1}{3}$. 
    \end{proof}
    
    Combining expressions (\ref{eq:1}) and (\ref{eq:2}) and Claim~\ref{clm:gen-lb-helper}, we have that with probability at least $\frac{2}{3}$,
        
    $$\varepsilon > \frac{1}{2^n} \left|\sum_{\oracle' \in D'_r} \CRAP'_{r}(\oracle') (\CRAP'_{r}(\oracle') - p'(\oracle'))\right| = \frac{1}{2^n} \left||D'_r| - \left(\sum_{\oracle' \in D'_r} \CRAP'_{r}(\oracle') p'(\oracle')\right)\right| = \frac{|D'_r|}{2^n} = 1.$$

    And so $\varepsilon > 1$ which contradicts our assumption. Thus $\CRAP'_{r}(\Oracle(r, x))$  cannot be approximated by a polynomial of degree at most $\frac{n}{4\log{m}} - 1$ with probability at least $\frac{2}{3}$ over the choice $r \from \R$ sampled uniformly at random. And therefore $\CRAP_{r}(\Oracle(r, x))$ cannot be $\eps$-approximated for every constant $\eps < \frac{1}{2}$ with a polynomial of degree less than $\frac{n}{4\log{m}}$ with probability at least $\frac{2}{3}$ over the choice $r \from \R$ sampled uniformly at random for any \emph{good base} $\R$ of size $m$.
\end{proof}

\subsection{Lower bound for ordered search}


Finally, we combine our general lower bound on the approximate degree of $\CRAP$ with the upper bound on approximating $\GT_i$ to conclude our lower bound on the approximate degree of ordered search. We will use the statement of Claim~\ref{clm:os2-ub} with a lower degree of polynomials approximating $\GT_i$ since, even though its proof is more complicated than the proof of the weaker bound, as it allows us to obtain a better lower bound on ordered search. 

First, we apply Theorem~\ref{thm:gen-lb} to obtain a lower bound on the approximate degree for $\CRAPWithParameter{\R_{\OS++}}$.


\begin{corollary}\label{cor:os2-lb}
    If $r \from \R_{\OS++}$ u.a.r. then with probability at least $\frac{2}{3}$ over the choice of $r$ every polynomial $p : \{0, 1\}^{m} \to \reals$ of degree at most $d = \frac{n}{16\log{n}} - 1$ fails to approximate $\CRAPWithParameter{\R_{\OS++}}_r$:

    $$\Pr_{r \from \R_{\OS++}}\left[\forall \varepsilon < \frac{1}{2}, \ \forall p, \deg(p) \leq d, \exists \oracle \in D[\R_{\OS++}]_r: \left|p(\oracle) - \CRAPWithParameter{\R_{\OS++}}_r(\oracle)\right| > \varepsilon\right] \geq \frac{2}{3}.$$
\end{corollary}
\begin{proof}
    The set $\R_{\OS++}$ is a \emph{good base} and has size $m = O(n^3 \log{n})$. By Theorem~\ref{thm:gen-lb}, with probability at least $\frac{2}{3}$ over the choice of $r$ every polynomial $p : \{0, 1\}^{m} \to \reals$ of degree at most $\frac{n}{4\log{m}} - 1$ fails to approximate $\CRAPWithParameter{\R_{\OS++}}_r$. But since the size of $\R_{\OS++}$ is $m \leq n^4$ for sufficiently large $n$ then every polynomial of degree at most $d = \frac{n}{16\log{n}} - 1 = \frac{n}{4\log{n^4}} - 1 \leq \frac{n}{4\log{m}} - 1$ fails to approximate $\CRAPWithParameter{\R_{\OS++}}_r$.
\end{proof}

By combining Claim~\ref{clm:os2-ub} and Corollary~\ref{cor:os2-lb}, we obtain the following.

\begin{theorem} \label{thm:os2-lb}
The approximate degree of ordered search is 

$$\adeg_{\frac{1}{2} - \gamma}(\OS_{2^n}) = \Omega\left(\frac{n}{\log^2{n}} - \log{\frac{1}{\gamma}}\right)$$

where $\gamma$ could depend on $n$, $0 < \gamma < \frac{1}{2}$.

\end{theorem}
\begin{proof}

Suppose $\OS_{2^n}$ can be $(\frac{1}{2} - \gamma)$-approximated by a bounded polynomial of degree $d$ for some $\frac{1}{2} > \gamma > 0$. By \cite[Theorem 1.1]{sherstov2012making}, for every $\delta > 0$, this polynomial can be converted to a polynomial $p$ of degree $O(d+\log{\frac{1}{\delta}})$ that $(\frac{1}{2} - \gamma +\delta)$-approximates $\OS_{2^n}$ and is robust to noise in its inputs. That is,
    
    $$|\OS_{N}(\oracle) - p(\oracle+\Delta)| < \frac{1}{2} - \gamma + \delta$$
    for all $\oracle \in \{0, 1\}^N$, all $\Delta \in [-\frac{1}{6}, \frac{1}{6}]^N$, and $N = 2^n$. If we put $\delta = \frac{\gamma}{2}$, then $p$ is a $(\frac{1}{2} - \frac{\gamma}{2})$-approximating polynomial for $\OS_{2^n}$ with degree $O\left(d + \log(\frac{1}{\gamma})\right)$.
    
    Note that $\OS_{2^n}(\GT_{0^n}(x), \GT_{0^{n-1}1}(x), \ldots, \GT_{1^n}(x)) = \CRAPWithParameter{\R_{\OS++}}_{r}(\Oracle(r, x))$ for every $r \in \R_{\OS++}$. So by Claim~\ref{clm:os2-ub}, there exists a constant $c$ such that the composed polynomial \\ $p(q_{(r, 0^n)}(\Oracle(r, x)), q_{(r, 0^{n-1}1)}(\Oracle(r, x)), \ldots q_{(r, 1^n)}(\Oracle(r, x)))$ has degree at most $\deg(p) \max_{i}(\deg(q_{(r, i)})) = c \left(d + \log(\frac{1}{\gamma})\right) \log{n}$ and approximates $\CRAPWithParameter{\R_{\OS++}}_{r}(\Oracle(r, x))$ to error $(\frac{1}{2} - \frac{\gamma}{2})$ with probability at least $\frac{2}{3}$ over the choice of $r \from \R_{\OS++}$. This holds because although the polynomials $q_{(r, i)}$ do not compute the functions $\GT_i$ exactly, but only approximate them with small error, the outer polynomial $p$ is robust to this small error in the inputs. Note also that while the composed polynomial is bounded on the domain of $\CRAP_r$, it may be arbitrarily unbounded on points outside its domain. 
    
    On the other hand, by Claim~\ref{cor:os2-lb}, with probability at least $\frac{2}{3}$ over the choice of $r$, the function $\CRAPWithParameter{\R_{\OS++}}_{r}$ cannot be approximated to any error $(\frac{1}{2} - \frac{\gamma}{2}) \in (0, \frac{1}{2})$ by a polynomial in $\Oracle$ of degree less than $\frac{n}{16\log{n}}$. By a union bound, with probability at least $1 - (1-\frac{2}{3}) - (1-\frac{2}{3}) = \frac{1}{3}$ both conditions on $r$ hold simultaneously. Thus there exists $r \in \R_{\OS++}$ such that $p(q_{(r, 0^n)}(\Oracle(r, x)), q_{(r, 0^{n-1}1)}(\Oracle(r, x)), \ldots q_{(r, 1^n)}(\Oracle(r, x)))$ approximates $\CRAP_{r}(\Oracle(r, x))$ and $\CRAP_{r}(\Oracle(r, x))$ cannot be approximated by a polynomial of degree less than $\frac{n}{16\log{n}}$. So
    
    $$c \left(d + \log(\frac{1}{\gamma})\right) \log{n} \geq \frac{n}{16\log{n}}.$$
    And thus
    
    $$d + \log(\frac{1}{\gamma}) \geq \frac{n}{16c\log^2{n}},$$
    so we conclude that
    $$d = \Omega\left(\frac{n}{\log^2{n}}  - \log(\frac{1}{\gamma})\right).$$

\end{proof}

\section{Anchored hidden string and hidden string}\label{sec:HS}

Now let us switch gears and consider the (anchored) hidden string problem, in which the goal is to reconstruct a string given information about the presence of specific substrings. In the decisional anchored hidden string ($\AHS$) problem, the information given as input consists of whether each string $s$ is a substring of the hidden input $x$ starting at position $i$, for all valid $i$ and $s$. The goal is then to compute $\parity(x)$.

In order to prove a lower bound for $\HSA$ we will follow the same outline as for $\OS$. That is, first we will introduce a convenient set $\R_{\AHS}$ of collections of $n$-bit strings and show that oracle $\OracleWithParameter{\R_{\AHS}}$ providing the inner products of $x$ with strings from the random sample from $\R_{\AHS}$ are useful for computing $\phi_{i, s}(x)$ for all possible queries $(i, s)$ where $i \in [n], s \in \{0, 1\}^{\leq n - i+1}$ and $\phi_{i, s}(x) = 1$ iff $s$ is a substring of $x$ starting from position $i$. After that, we will show that it is hard to compute $\CRAPWithParameter{\R_{\AHS}}$, the parity of $x$ using the oracle $\OracleWithParameter{\R_{\AHS}}$ with high probability. And finally, we will conclude that computing $\AHS$ is hard since composing an approximating polynomial for $\AHS$ with polynomials approximating $\phi_{i, s}(x)$ would allow us to approximate the $\CRAPWithParameter{\R_{\AHS}}$ function.


\subsection{Preliminaries}

W define several functions in order to formalize the problem.

Throughout this section, we use the following notation $\{0, 1\}^{\leq n}$ to denote the set of all bit strings of size at most $n$: $\{0, 1\}^{\leq n} = \bigcup_{k = 0}^{n} \{0, 1\}^k$.

\begin{definition}
    For all $s \in \{0, 1\}^{\leq n}$ define the function $\chi_s: \{0, 1\}^n \to \{0, 1\}$ to be the indicator of whether the input string $x$ has $s$ as a substring: that is, there exists an integer $i$ such that $x_{i + k - 1} = s_k$ for all $1 \leq k \leq |s|$ then $\chi_s(x) = 1$ and otherwise $\chi_s(x) = 0$.
\end{definition}

\begin{definition}
    Define the ``hidden string'' function $\HS_N: \{0, 1\}^N \to \{0, 1\}$ be the partial function that takes $N = |\{0, 1\}^{\leq n}| = 2^{n+1}-1$ inputs, each corresponding to a substring $s \in \{0, 1\}^{\leq n}$, and, given a collection of $\chi_s(x)$ for some fixed $x \in \{0, 1\}^n$ as an input, outputs $\parity(x)$. 
\end{definition}

We will also consider a variation of this problem where the additional information is not only whether a specific substring is present in the hidden string, but if this substring is present at a specific location of the hidden string.

\begin{definition}
    For all $i \in [n]$ and $s \in \{0, 1\}^{\leq n-i+1}$ define the function $\phi_{i, s}: \{0, 1\}^n \to \{0, 1\}$ to be the indicator of whether the input string $x$ has $s$ as a substring starting from position $i$: that is, if $x_{i + k - 1} = s_k$ for all $1 \leq k \leq |s|$ then $\phi_{i, s}(x) = 1$ and otherwise $\phi_{i, s}(x) = 0$.
\end{definition}

\begin{definition}
    Let the ``anchored hidden string'' function $\AHS_N : \{0, 1\}^N \to \{0, 1\}$ be the partial function that takes $N = |\{(i, s) \mid i \in [n], s \in \{0, 1\}^{\leq n-i+1}\}| = 2^{n+2}-n-4$ inputs, each corresponding to a pair of $i \in [n]$ and $s \in \{0, 1\}^{\leq n-i+1}$, and, given a collection of $\phi_{i, s}(x)$ for some fixed $x \in \{0, 1\}^n$ as an input, outputs $\parity(x)$. 
\end{definition}

\subsection{Approximating polynomials for $\phi_{i, s}$}

We start our proof by showing that for some \emph{good base} $\R_{\AHS}$ the oracle $\OracleWithParameter{\R_{\AHS}}$ could be used to make the computation of the functions $\phi_{i, s}$ more efficient.

\begin{claim}\label{clm:ahs-ub}
   There exists a \emph{good base} $\R_{\AHS}$ of size $m = O(n^3)$ such that if $r \from \R_{\AHS}$ is sampled uniformly at random, then with probability at least $\frac{2}{3}$ over the choice of $r$ there exists a family of polynomials $\{q_{(r, i, s)}: \{0, 1\}^{m} \to \{0, 1\} \mid i \in [n], s \in \{0, 1\}^{\leq n-i+1}\}$ of degree at most $4$ such that given $\OracleWithParameter{\R_{\AHS}}(r, x)$ as the input, each polynomial $q_{(r, i, s)}(\OracleWithParameter{\R_{\AHS}}(r, x)$ approximates the corresponding $\phi_{i, s}(x)$ with error at most $\frac{1}{6}$. That is,
   
    $$\Pr_{r \from \R_{\AHS}}\left[\exists i \in [n], x \in \{0, 1\}^n, s \in \{0, 1\}^{\leq n-i+1}: \left|q_{(r, i, s)}(\OracleWithParameter{\R_{\AHS}}(r, x)) - \phi_{i, s}(x)\right| > \frac{1}{6}\right] < \frac{1}{3}.$$
    
\end{claim}
\begin{proof}

The proof of this statement follows the same outline as the proof of  Claim~\ref{clm:os-ub}. First, we will construct a \emph{good base} $\R_{\AHS}$, and then we will show (in two stages) how to compute $\phi_{i, s}$ given $\OracleWithParameter{\R_{\AHS}}$.

\textbf{Constructing the \emph{good base} $\R_{\AHS}$.}
Our base for the oracle should contain all the strings needed to check the equality with every substring of $x$, so let $\hat{\R} = \R^{1^{n}} \times \left(\R^{1^{n-1} 0^1} \times \R^{0^{1} 1^{n-1}}\right) \times \left(\R^{1^{n-2} 0^2} \times \R^{0^{1} 1^{n-1} 0^1} \times \R^{0^{2} 1^{n-2}}\right) \times \ldots \times \left(\bigtimes_{i = 0}^{k}\R^{0^{i} 1^{i+k}0^{n-i-k}}\right) \times \ldots \times \left(\bigtimes_{i = 0}^{n-1}\R^{0^{i} 1^{i+1}0^{n-(i+1)}}\right)$. See Figure~\ref{fig:struct-ahs} for an illustration.

\begin{figure}[h]
\begin{center}
\begin{tabular}{ccc}
{\Large $\tau$} & {\Large$\R^{\tau}$} & Structure of $\R^{\tau}$\\ \\
{\Large 1111} & {\Large $\{0, 1\}^4$}   & \begin{tabular}{|c|c|c|c|} \hline \cellcolor{blue_color}$\star$ & \cellcolor{blue_color}$\star$ & \cellcolor{blue_color}$\star$ & \cellcolor{blue_color}$\star$  \\ \hline \end{tabular} \\
                         &                                                                                                                                                                                                                                           \\
{\Large 1110} & {\Large $\{0, 1\}^3 \{0\}$} & \begin{tabular}[c]{@{}l@{}}\begin{tabular}{|c|c|c|c|} \hline \cellcolor{blue_color}$\star$ & \cellcolor{blue_color}$\star$ & \cellcolor{blue_color}$\star$ & $0$ \\ \hline \end{tabular}\end{tabular}                     \\
{\Large 0111} & {\Large $\{0\} \{0, 1\}^3$}   & \begin{tabular}{|c|c|c|c|c|}\hline $0$ & \cellcolor{blue_color}$\star$ & \cellcolor{blue_color}$\star$ & \cellcolor{blue_color}$\star$ \\ \hline \end{tabular}                     \\
                         &                                                                                                                                                                                                                                           \\
{\Large 1100} & {\Large $\{0, 1\}^2 \{0\}^2$} & \begin{tabular}{|c|c|c|c|} \hline \cellcolor{blue_color}$\star$ & \cellcolor{blue_color}$\star$ & $0$ & $0$ \\ \hline \end{tabular}                                         \\
{\Large 0110} & {\Large $\{0\} \{0, 1\}^2 \{0\}$} & \begin{tabular}{|c|c|c|c|} \hline $0$ & \cellcolor{blue_color}$\star$ & \cellcolor{blue_color}$\star$ & $0$ \\ \hline \end{tabular}                                         \\
{\Large 0011} & {\Large $\{0\}^2 \{0, 1\}^2$}   & \begin{tabular}{|c|c|c|c|} \hline $0$ & $0$ & \cellcolor{blue_color}$\star$ & \cellcolor{blue_color}$\star$ \\ \hline \end{tabular}                                         \\
                         &                                                                                                                                                                                                                                           \\
{\Large 1000} & {\Large $\{0, 1\} \{0\}^3$} & \begin{tabular}{|c|c|c|c|c|} \hline \cellcolor{blue_color}$\star$ & $0$ & $0$ & $0$ \\ \hline \end{tabular}                                                             \\
{\Large 0100} & {\Large $\{0\} \{0, 1\} \{0\}^2$} & \begin{tabular}{|c|c|c|c|c|} \hline $0$ & \cellcolor{blue_color}$\star$ & $0$ & $0$ \\ \hline \end{tabular}                                                             \\
{\Large 0010} & {\Large $\{0\}^2 \{0, 1\} \{0\}$} & \begin{tabular}{|c|c|c|c|} \hline $0$ & $0$ & \cellcolor{blue_color}$\star$ & $0$ \\ \hline \end{tabular}                                                             \\
{\Large 0001} & {\Large $\{0\}^3 \{0, 1\}$}   & \begin{tabular}{|c|c|c|c|c|} \hline $0$ & $0$ & $0$ & \cellcolor{blue_color}$\star$ \\ \hline \end{tabular}                                                             \\
\end{tabular}
\end{center}
\caption{Structure of $\hat{\R}$ for $n = 4$. Blue cells with $\star$ represent places where either 0 or 1 values could be.}
\label{fig:struct-ahs}
\end{figure}

On top of this, we are going to add two other steps to the structure. First, we are going to have $t$ individual ``prepackaged'' copies. Let $\R_1 = \ldots = \R_{t} = \times_{\alpha} \hat{\R}$, and $\R' = \bigtimes_{j \in [t]} \R_{j}$. Secondly, we add a set of ``basis'' strings to the structure: $\mathcal{B} = \{\one_1\} \times \{\one_2\} \times \ldots \times \{\one_{n-1}\} \times \{\one_{n}\}= \{10\ldots0\}\times\{010\ldots0\} \times \ldots \times \{00\ldots010\} \times \{00\ldots01\}$. The final underlying structure of the oracle will be a Cartesian product of all $t$ copies and $\mathcal{B}$: 
$\R = \mathcal{B} \times \R' = \mathcal{B} \times (\bigtimes_{j \in [t]} \R_{j})$. See Figure~\ref{fig:oracle-ahs} for an illustration. We also set the parameters to be $\alpha = 4, t = 1000n\ln{2}$.
Notice that this set $\R_{\AHS}$ is a \emph{good base} by construction, and has size $m = n + \alpha t \frac{n(n+1)}{2} =  O(n^3)$.

\begin{figure}
\begin{center}
\begin{tabular}{ccccccccccc}
$\mathcal{B}$ & & $\R_1$ & & $\R_2$ & & & & $\R_{t-1}$ &  & $\R_t$ \\
 & & $\alpha$ copies & & $\alpha$ copies & & & & $\alpha$ copies &  & $\alpha$ copies \\
\footnotesize
\tabcolsep=0.11cm
\begin{tabular}{|cccc|} \hline \cellcolor{yellow}1 & 0 & 0 & 0  \\ 0 & \cellcolor{yellow}1 & 0 & 0 \\ 0 & 0 & \cellcolor{yellow}1 & 0 \\ 0 & 0 & 0 & \cellcolor{yellow}1 \\ \hline \end{tabular} & $\bigtimes$ & \begin{tabular}{|cccc|} \hline $\cellcolor{blue_color}$ & $\cellcolor{blue_color}$ & $\cellcolor{blue_color}$ & $\cellcolor{blue_color}$ \\  $\cellcolor{blue_color}$ & $\cellcolor{blue_color}$ & $\cellcolor{blue_color}$ &                    \\                    & $\cellcolor{blue_color}$ & $\cellcolor{blue_color}$ & $\cellcolor{blue_color}$ \\ $\cellcolor{blue_color}$ & $\cellcolor{blue_color}$ &                    &                    \\                    & $\cellcolor{blue_color}$ & $\cellcolor{blue_color}$ &                    \\                    &                    &  $\cellcolor{blue_color}$ & $\cellcolor{blue_color}$ \\ $\cellcolor{blue_color}$ &                    &                    &                    \\                    & $\cellcolor{blue_color}$ &                    &                    \\                    &                    & $\cellcolor{blue_color}$ &                    \\                    &                    &                    & $\cellcolor{blue_color}$ \\ \hline \end{tabular} & $\bigtimes$ & \begin{tabular}{|cccc|} \hline $\cellcolor{blue_color}$ & $\cellcolor{blue_color}$ & $\cellcolor{blue_color}$ & $\cellcolor{blue_color}$ \\  $\cellcolor{blue_color}$ & $\cellcolor{blue_color}$ & $\cellcolor{blue_color}$ &                    \\                    & $\cellcolor{blue_color}$ & $\cellcolor{blue_color}$ & $\cellcolor{blue_color}$ \\ $\cellcolor{blue_color}$ & $\cellcolor{blue_color}$ &                    &                    \\                    & $\cellcolor{blue_color}$ & $\cellcolor{blue_color}$ &                    \\                    &                    &  $\cellcolor{blue_color}$ & $\cellcolor{blue_color}$ \\ $\cellcolor{blue_color}$ &                    &                    &                    \\                    & $\cellcolor{blue_color}$ &                    &                    \\                    &                    & $\cellcolor{blue_color}$ &                    \\                    &                    &                    & $\cellcolor{blue_color}$ \\ \hline \end{tabular} & $\bigtimes$ & $\ldots$ & $\bigtimes$ & \begin{tabular}{|cccc|} \hline $\cellcolor{blue_color}$ & $\cellcolor{blue_color}$ & $\cellcolor{blue_color}$ & $\cellcolor{blue_color}$ \\  $\cellcolor{blue_color}$ & $\cellcolor{blue_color}$ & $\cellcolor{blue_color}$ &                    \\                    & $\cellcolor{blue_color}$ & $\cellcolor{blue_color}$ & $\cellcolor{blue_color}$ \\ $\cellcolor{blue_color}$ & $\cellcolor{blue_color}$ &                    &                    \\                    & $\cellcolor{blue_color}$ & $\cellcolor{blue_color}$ &                    \\                    &                    &  $\cellcolor{blue_color}$ & $\cellcolor{blue_color}$ \\ $\cellcolor{blue_color}$ &                    &                    &                    \\                    & $\cellcolor{blue_color}$ &                    &                    \\                    &                    & $\cellcolor{blue_color}$ &                    \\                    &                    &                    & $\cellcolor{blue_color}$ \\ \hline \end{tabular} & $\bigtimes$ & \begin{tabular}{|cccc|} \hline $\cellcolor{blue_color}$ & $\cellcolor{blue_color}$ & $\cellcolor{blue_color}$ & $\cellcolor{blue_color}$ \\  $\cellcolor{blue_color}$ & $\cellcolor{blue_color}$ & $\cellcolor{blue_color}$ &                    \\                    & $\cellcolor{blue_color}$ & $\cellcolor{blue_color}$ & $\cellcolor{blue_color}$ \\ $\cellcolor{blue_color}$ & $\cellcolor{blue_color}$ &                    &                    \\                    & $\cellcolor{blue_color}$ & $\cellcolor{blue_color}$ &                    \\                    &                    &  $\cellcolor{blue_color}$ & $\cellcolor{blue_color}$ \\ $\cellcolor{blue_color}$ &                    &                    &                    \\                    & $\cellcolor{blue_color}$ &                    &                    \\                    &                    & $\cellcolor{blue_color}$ &                    \\                    &                    &                    & $\cellcolor{blue_color}$ \\ \hline \end{tabular}
\end{tabular}
\end{center}
\caption{Structure of $\R_{\AHS}$ for $n = 4$. Each $\R_j$ consist of $\alpha$ copies of $\hat{\R}$.}
\label{fig:oracle-ahs}
\end{figure}

\normalsize

\textbf{Constructing the family of approximating polynomials.} 
For all $i \in [n], s \in \{0, 1\}^{\leq n-i+1}, j \in [t], r \in \R_{\AHS}$ let $B_{(r, i, s, j)}(\Oracle[\R_{\AHS}](r, x))$ be the following deterministic algorithm.

    \begin{enumerate}
        \item Set $\tau = 0^{i} 1^{|s|} 0^{n - |s| - i}$
        \item Set $s' = 0^{i} s 0^{n - |s| - i}$ so $s'$ is an $n$-bit string
        \item For all $v \in [m]$ such that $v$ corresponds to $n$-bit strings drawn from $\R^{\tau}$ within the $j$th copy $\R_j$:
        \item \qquad Compute $\langle s', r_v \rangle$ and compare it to $(\Oracle(r, x))_v = \langle x, r_v \rangle$. 
        \item If for some $v$ the inner products don't have the same value, $\langle i, r_v \rangle \neq (\Oracle(r, x))_v$, then reject. 
        \item Otherwise, accept.
    \end{enumerate}

    This algorithm determines if $s$ equal to the substring of $x$ of length $|s|$ that starts at the $i$th position with probability at least $1 - 2^{-\alpha} = 1 - 2^{-4} \geq \frac{11}{12}$ independently of the choice of $j \in [t]$. That is, for all $j \in [t]$ and for all $i \in [n], s \in \{0, 1\}^{n-i}, x \in \{0, 1\}^n$

    $$\Pr_{r \from \R_{\AHS}}[B_{(r, i, s, j)}(\Oracle(r, x)) = \phi_{i, s}(x)] \geq \frac{11}{12}.$$

    This algorithm makes at most $\alpha = 4$ queries to the oracle $\Oracle(r, x)$. 

    We have shown that for every fixed $i \in [n], s \in \{0, 1\}^{\leq n-i+1}$, and $x \in \{0, 1\}^n$ there are many $r \in \R_{\AHS}$ that if used as the first input for the oracle $\Oracle$ allow $B_{(r, i, s, j)}$ to compute $\phi_{i, s}(x)$. 
    
    Let $W(i, s, x, r, j)$ be the indicator that the $j$-th ``package'' of random strings in $r$ defines a set of ``bad'' random strings for $i, s, x$: $W(i, s, x, r, j) = 1$ if and only if $B_{(r, i, s, j)}(\Oracle(r, x)) \neq \phi_{i, s}(x)$. We established that $B_{(r, i, s, j)}(\Oracle(r, x))$ works well if given a random $r \from \R_{\AHS}$ for every $j \in [t]$ and the probability of this algorithm outputting an incorrect answer is at most $\frac{1}{12}$. So $$\Pr_{r \from \R_{\AHS}}[W(i, s, x, r, j) = 1] = \E_{r \from \R_{\AHS}}[W(i, s, x, r, j)] \leq \frac{1}{12}.$$

    Using the same transformation from before (described in detail in the proof of Claim~\ref{clm:os-ub}), we design a new family of algorithms that succeeds on all $i, s, x$ simultaneously. For all $i \in [n], s \in \{0, 1\}^{\leq n-i+1}, r \in \R_{\OS++}$ let $A_{(r, i, s)}(\Oracle(r, x))$ be the following randomized algorithm:
        
    \begin{itemize}
        \item Choose $j \from [t]$ uniformly at random.
        \item Run $B_{(r, i, s, j)}(\Oracle(r, x))$.
    \end{itemize}
    The number of queries that $A_{(r, i, s)}$ makes to the oracle is the same as $B_{(r, i, s, j)}$ which is $\alpha = 4$. We fix arbitrary $i, s, x$ and evaluate the following probability.
    
    $$ \Pr_{r \from \R_{\AHS}}\left[\Pr_{j \from [t]}[B_{(r, i, s,  j)} \neq \phi_{i, s}(x)] > \frac{1}{6}\right] = \Pr_{r \from \R_{\AHS}}\left[\frac{1}{t} \sum_{j \in [t]} W(i, s, x, r, j) > \frac{1}{6}\right].$$
    We established that $\E_{r \from \R_{\AHS}}[W(i, s, x, r, j)] \leq \frac{1}{12}$ and so by Hoeffding's inequality,
    
    $$\Pr_{r \from \R_{\AHS}}\left[\frac{1}{t} \sum_{j \in [t]} W(i, s, x, r, j) > \frac{1}{12} + \frac{1}{12}\right] \leq e^{-2\frac{t}{144}} \leq 2^{-\frac{2000n}{144}}.$$
    And by a union bound,
    
    $$\Pr_{r \from \R_{\AHS}}\left[\exists i \in [n], s \in \{0, 1\}^{\leq n - i+1}, x \in \{0, 1\}^n: \frac{1}{t} \sum_{j \in [t]} W(i, s, x, r, j) > \frac{1}{6}\right] \leq 2^{2n+2} 2^{-\frac{2000n}{144}} < \frac{1}{3},$$
    since the number of possible pairs of $i \in [n]$ and $s \in \{0, 1\}^{\leq n - i+1}$ is $\sum_{i = 1}^{n} \sum_{|s| = 1}^{n-i+1} 2^{|s|} \leq 2^{n+2}$.
    So, we proved that
        
    $$\Pr_{r \from \R_{\AHS}}\left[\exists i \in [n], s \in \{0, 1\}^{\leq n-i+1}, x \in \{0, 1\}^n: \Pr_{j \from [t]}[A_{(r, i, s)}(\Oracle(r, x))) \neq \phi_{i, s}(x)] > \frac{1}{6}\right] < \frac{1}{3}.$$
        
    The last step is to convert this family of query algorithms into a family of approximating polynomials. Let $q_{(r, i, s)}$ denote the acceptance probability of $A_{(r, i, s)}$ which is a polynomial of degree at most $\alpha = 4$ such that
        
    $$\Pr_{r \from \R_{\AHS}}\left[\exists i \in [n], s \in \{0, 1\}^{\leq n-i+1}, x \in \{0, 1\}^n: \left|q_{(r, i, s)}(\OracleWithParameter{\R_{\AHS}}(r, x)) - \phi_{i, s}(x)\right| > \frac{1}{6}\right] < \frac{1}{3}$$
    which is exactly what we were looking for.

\end{proof}

\subsection{Lower bounds for anchored hidden string and hidden string}


Following the same framework, we can combine the general statement about the degree of $\CRAP$ proven earlier with the upper bound for approximating $\phi_{i, s}$ to conclude our lower bound on the approximate degree of anchored hidden string.

First, we apply Theorem~\ref{thm:gen-lb} to obtain a lower bound on approximate degree for $\CRAPWithParameter{\R_{\AHS}}$.


\begin{corollary}\label{cor:ahs-lb}
    If $r \from \R_{\AHS}$ u.a.r. then with probability at least $\frac{2}{3}$ over the choice of $r$ every polynomial $p : \{0, 1\}^{m} \to \reals$ of degree at most $d = \frac{n}{16\log{n}} - 1$ fails to approximate $\CRAPWithParameter{\R_{\AHS}}_r$:

    $$\Pr_{r \from \R_{\AHS}}\left[\forall \varepsilon < \frac{1}{2}, \ \forall p, \deg(p) \leq d, \exists \oracle \in D[\R_{\AHS}]_r: \left|p(\oracle) - \CRAPWithParameter{\R_{\AHS}}_r(\oracle)\right| > \varepsilon\right] \geq \frac{2}{3}.$$
\end{corollary}
\begin{proof}
    The set $\R_{\AHS}$ is a \emph{good base} and has size $m = O(n^3)$. By Theorem~\ref{thm:gen-lb}, with probability at least $\frac{2}{3}$ over the choice of $r$ every polynomial $p : \{0, 1\}^{m} \to \reals$ of degree at most $\frac{n}{4\log{m}} - 1$ fails to approximate $\CRAPWithParameter{\R_{\AHS}}_r$. But since the size of $\R_{\AHS}$ is $m \leq n^4$ for sufficiently large $n$ then every polynomial of degree at most $d = \frac{n}{16\log{n}} - 1 = \frac{n}{4\log{n^4}} - 1 \leq \frac{n}{4\log{m}} - 1$ fails to approximate $\CRAPWithParameter{\R_{\AHS}}_r$.
\end{proof}

By combining Claim~\ref{clm:ahs-ub} and Corollary~\ref{cor:ahs-lb}, we obtain the following.

\begin{theorem} \label{thm:ahs-lb}
The approximate degree of $\AHS_{N}$ is

$$\adeg_{\frac{1}{2} - \gamma}(\AHS_{N}) = \Omega\left(\frac{n}{\log{n}} - \log{\frac{1}{\gamma}}\right)$$
where $N = 2^{n+2}-n-4$ and $\gamma$ could depend on $n$, $0 < \gamma < \frac{1}{2}$.

\end{theorem}
\begin{proof}
Suppose $\AHS_{N}$ can be $(\frac{1}{2} - \gamma)$-approximated by a bounded polynomial of degree $d$.

By the same argument as used in Theorem~\ref{thm:os2-lb} we can conclude that there exists a polynomial $p$ of degree $O(d + \log{1}{\gamma})$ that $(\frac{1}{2} - \frac{\gamma}{2})$-approximates $\AHS_N$ and is robust to noise. That is,

    $$|\AHS_{N}(\oracle) - p(\oracle+\Delta)| < \frac{1}{2} - \frac{\gamma}{2}$$
    for all $\oracle \in \{0, 1\}^N$ where $\Delta \in [-\frac{1}{6}, \frac{1}{6}]^N$. 
    
    Note that $\AHS_{N}((\phi_{i, s}(x))_{i \in [n], s \in \{0, 1\}^{\leq n-i+1}}) = \CRAPWithParameter{\R_{\AHS}}_{r}(\Oracle(r, x))$ for every $r \in \R_{\AHS}$. So by Claim~\ref{clm:ahs-ub} the polynomial $p(q_{(r, i, s)}(\Oracle(r, x)))$ of degree at most $\deg(p) \max_{i, s}(\deg(q_{(r, i, s)})) = 4c(d+\log{\frac{1}{\gamma}})$ for some constant $c$ approximates $\CRAPWithParameter{\R_{\AHS}}_{r}(\Oracle(r, x))$ to error $(\frac{1}{2} - \frac{\gamma}{2})$ with probability at least $\frac{2}{3}$ over the choice of $r \from \R_{\AHS}$. This holds because although the polynomials $q_{(r, i, s)}$ do not compute the functions $\phi_{i, s}$ exactly, but only approximate them with small error, the outer polynomial $p$ is robust to this small error in the inputs. Note also that while the composed polynomial is bounded on the domain of $\CRAP_r$, it may be arbitrarily unbounded on points outside its domain. 
    
    On the other hand, by Claim~\ref{cor:ahs-lb}, with probability at least $\frac{2}{3}$ over the choice of $r$, the function $\CRAPWithParameter{\R_{\AHS}}_{r}$ cannot be approximated by a polynomial in $\Oracle$ of degree less than $\frac{n}{16\log{n}}$. By a union bound, with probability at least $(1 - (1-\frac{2}{3}) - (1-\frac{2}{3})) = \frac{1}{3}$ both conditions on $r$ hold simultaneously. Thus there exists $r \in \R_{\AHS}$ such that $p(q_{(r, i, s)}(\Oracle(r, x)))$ approximates $\CRAP_{r}(\Oracle(r, x))$ and $\CRAP_{r}(\Oracle(r, x))$ cannot be approximated by a polynomial of degree less than $\frac{n}{16\log{n}}$. So
    
    $$4 c(d + \log{\frac{1}{\gamma}}) \geq \frac{n}{16\log{n}}.$$
    And thus
    
    $$d = \Omega\left(\frac{n}{\log{n}} - \log{\frac{1}{\gamma}}\right).$$
\end{proof}

This lower bound on the approximate degree of $\AHS_N$ entails a lower bound on the approximate degree of $\HS_N$. In \cite{cleve2012reconstructing} the authors gave a reduction between these two problems, showing that if there exists a quantum query algorithm for $\HS$ then there exists a quantum query algorithm for $\AHS$ with a small blow-up in the number of queries. Specifically, they showed how to compute $\AHS$ by applying a query algorithm for $\HS$ with a slightly bigger input, where each bit of the bigger input can be computed using a constant number of queries to the original $\AHS$ input. This argument works just as well to relate the approximate degrees of $\HS$ and $\AHS$, giving the following statement.

\begin{claim}\label{clm:ahs-hs}
    If for every $n'$, there is a polynomial of degree $d(n')$ approximating $\HS_{2^{n'+1} - 1}$ to some error, then for every $n$ there is a polynomial of degree $2d(10n \log{n})$ approximating $\AHS_{2^{n+2} - n - 4}$ to the same error.
\end{claim}

This allows us to prove a lower bound for $\HS_N$ as well.

\begin{corollary}\label{cor:hs-lb}
The approximate degree of $\HS_{N}$ is

$$\adeg_{\frac{1}{2} - \gamma}(\HS_{N}) = \Omega\left(\frac{n}{\log^2{n}} - \log(\frac{1}{\gamma})\right)$$
where $N = 2^{n+1}-1$ and $\gamma$ could depend on $n$, $0 < \gamma < \frac{1}{2}$.
\end{corollary}

\begin{proof}
    By Claim~\ref{clm:ahs-hs} if there exists a polynomial of degree $d(n')$ approximating $\HS_{2^{n'+1}-1}$ then there exists a polynomial of degree $2d(10 n \log{n})$ approximating $\AHS_{2^{n+2}-n-2}$. On the other hand, by Corollary~\ref{cor:ahs-lb} no polynomial of degree less than $\frac{cn}{\log{n}} - c\log(\frac{1}{\gamma})$ can approximate $\AHS_{2^{n+2}-n-2}$ to error $\frac{1}{2} - \gamma$ for some constant $c$. Therefore,

    $$2d(10n\log{n}) \geq \frac{cn}{\log{n}} - c\log{\frac{1}{\gamma}}.$$
    Set $n' = 10 n \log{n}$. Then

    $$d(n') \geq \frac{cn}{2\log{n}} - c\log{\frac{1}{\gamma}} = \frac{cn'}{20 \log^2{n}} - c\log{\frac{1}{\gamma}} \geq \frac{cn'}{20 \log^2{n'}} - c\log{\frac{1}{\gamma}}.$$
    And thus 
    $$\adeg_{\frac{1}{2} - \gamma}(\HS_{2^{n'+1}-1}) = \Omega\left(\frac{n'}{\log^2{n'}} - c\log{\frac{1}{\gamma}}\right).$$

\end{proof}

\paragraph{Acknowledgments.} We thank Arkadev Chattopadhyay for suggesting the problem of determining the approximate degree of ordered search, and Arkadev and Justin Thaler for many helpful conversations about it. We also thank the anonymous TQC 2023 reviewers for helpful suggestions on the presentation.






\bibliographystyle{alpha}
\bibliography{references}

\begin{appendix}
\section{Upper bounds for the unbounded error regime} \label{sec:unbounded-ub}

We describe upper bounds on the randomized (and hence, quantum) query complexities of the ordered search and hidden string problems in the setting of weakly unbounded error. Our algorithms are simple modifications of the corresponding deterministic algorithms. They show that the approximate degree and quantum query lower bounds we prove for these problems are nearly tight.

\subsection{Ordered search}

\paragraph{Reconstruction.}
We first describe a randomized query algorithm that computes ordered search (reconstruction version) with probability at least $\gamma$ while making $O(n - \log{\frac{1}{\gamma}})$ queries.

To attempt to identify a hidden string $x$, the algorithm makes the first $t$ queries of binary search and thus exactly identifies the first $t$ bits of $x$. Then it samples the rest of the bits uniformly at random and outputs the resulting $n$-bit string. It succeeds in sampling the correct sequence of the last $(n-t)$ bits with probability at least $2^{-(n-t)}$. The upper bound follows by setting $t = n - \log{\frac{1}{\gamma}}$.

\paragraph{Decision.}
We now modify the algorithm above to compute the decision version of ordered search with probability at least $\frac{1}{2}+\gamma$, while making $O(n - \log{\frac{1}{\gamma}})$ queries.

The algorithm makes the first $t$ queries of binary search to exactly identify the first $t$ bits of $x$. Then it samples the rest of the bits uniformly at random, obtaining an $n$-bit candidate string $x'$. It then queries the input twice, on indices $x' - 1$ and $x'$, to check that $x \not \le x' - 1$ and $x \le x'$. If both conditions hold, then $x' = x$ and the algorithm outputs outputs $\parity(x') = \parity(x)$. Otherwise, it outputs a random bit.

This algorithm succeeds when either it succeeds at identifying $x$ (which happens with probability $2^{-(n-t)}$) or when it fails to identify $x$, but correctly guesses its parity, which happens with probability $\frac{1}{2} (1-2^{-(n-t)})$. Thus the algorithm succeeds with probability at least $\frac{1}{2}+2^{-(n-t)-1}$. The upper bound follows by setting $t = n + 1 - \log{\frac{1}{\gamma}}$.

\subsection{Hidden string}

Since our unbounded-error algorithm relies on the deterministic algorithm for the hidden string, we provide a sketch of that algorithm here.

\begin{theorem*}(\cite{SkienaS95})
There exists a deterministic query algorithm that computes hidden string (reconstruction version) using $O(n)$ queries.
\end{theorem*}
Here is the sketch of the algorithm:
    \begin{enumerate}
        \item Let $s$ be the empty string.
        \item While either $s0$ or $s1$ is present in $x$:
        \item \qquad Update $s$ to be either $s0$ or $s1$, whichever is present
        \item While either $0s$ or $1s$ is present in $x$:
        \item \qquad Update $s$ to be either $0s$ or $1s$, whichever is present
    \end{enumerate}
    This algorithm makes at most $2n+2$ queries to the input and outputs the hidden string $x$.

\paragraph{Reconstruction.}
We now describe a randomized query algorithm that computes hidden string (reconstruction version) with probability at least $\gamma$ while making $O(n - \log{\frac{1}{\gamma}})$ queries.

The algorithm makes the first $t$ steps of the exact deterministic algorithm above and thus exactly identifies $t$ bits of $x$. Then it samples the rest of the bits uniformly at random, samples a location among these $n-t$ bits in which to insert the identified substring of $x$, and outputs the resulting $n$-bit string. It succeeds in guessing the location for the substring with probability at least $\frac{1}{n-t+1}$ and it succeeds in guessing the correct values of the rest of the bits with probability at least $2^{-(n-t)}$. The upper bound follows by setting $t = 2n - \log{\frac{1}{\gamma}}$.

\paragraph{Decision.}
We now describe a randomized query algorithm that computes hidden string (decision version) with probability at least $\frac{1}{2}+\gamma$ while making $O(n - \log{\frac{1}{\gamma}})$ queries.

The algorithm makes the first $t$ steps of the exact deterministic algorithm above and thus exactly identifies $t$ bits of $x$. Then it samples the rest of the bits uniformly at random, samples where to split this string of random bits to put the identified substring of $x$, thus getting an $n$-bit string $x'$. It then queries if $x'$ is a substring of $x$. Since $|x'| = |x|$, this is equivalent to checking if $x' = x$. If the answer is ``yes'', then the algorithm succeeded in identifying $x$ and it outputs $\parity(x') = \parity(x)$. Otherwise, it samples a random bit and outputs it. 

This algorithm succeeds when either it successfully identifies $x$ (which happens with probability $\frac{1}{n-t}2^{-(n-t)}$) or when it fails to identify $x$ but correctly guesses the value of its parity, which happens with probability $\frac{1}{2} (1-\frac{1}{n-t}2^{-(n-t)})$. Thus the algorithm succeeds with probability at least $\frac{1}{2}+\frac{1}{n-t}2^{-(n-t)-1}$. The upper bound follows by taking $t = 2n + 1 - \log{\frac{1}{\gamma}}$.

\section{Proof of Claim \ref{clm:os2-ub}}\label{sec:os2-ub}

    The communication protocol that we use to construct our polynomials itself is based on a result by \cite{feige1994computing} on algorithms in a noisy comparison model. To understand the protocol and to convert it to the family of polynomials later we need to open up the protocol and state their result.

    The following is implicit in \cite{feige1994computing}:
    \begin{claim}\label{clm:comp-tree}
        Consider the following problem. There is an unknown ``key'' in $[n]$. Given the ability to ask questions of the type ``is the unknown key greater than $a$?'' for every $a \in [n]$ and get the correct answer with probability at least $\frac{3}{4}$ independently for each question, the algorithm's goal is to  find the correct location in $(0, n]$ for the key while minimizing the number of questions it asks.
        
        Then there exists an algorithm that finds the correct location in $(0, n]$ for the key by asking at most $c\log{n}$ questions for some constant $c$ with probability at least $\frac{11}{12}$. 
    \end{claim}

    The algorithm basically performs a binary search for the correct location of the key with slight modifications. We are not going to present the algorithm in detail, but we are going to describe what questions the algorithm asks along the way. For each interval $(a, b]$ where $b-a > 1$, the algorithm seeks answers to the following questions: ``is the key $> a$?'', ``is the key $> b$?'' and ``is the key $> \frac{a+b}{2}$?'', each correct with probability $\frac{3}{4}$ independently from all other questions, and the algorithm needs to be able to ask each potential question of each type $c\log{n}$ times for some constant $c$. For each interval $(a, a+1]$ the algorithm seeks answers to the following questions: ``is the key $> a$?'' and ``is the key $> a+1$?'', each correct with probability at least $\frac{3}{4}$ independently from all other questions, and the algorithm needs to be able to ask each potential question of each type $2c^2\log^2{n}$ times. Note that if the range for the key position is $(0, n]$ then questions ``is the key $> 0$?'' and ``is the key $ > n$?'' can be omitted by the algorithm since it already knows the correct answer to them.

\begin{figure}[h]
\begin{center}
\begin{tabular}{cccc}
Interval                                       & Questions                        & {\Large $\tau$}              & Structure of $\R^{\tau}$ \\
                                               &                                  &                     &                          \\
$(0, n]$                                       & ``is the key $> \frac{n}{2}$?''  & $1^{n/2} 0^{n/2}$   & 
    \small
    \tabcolsep=0.12cm
    \begin{tabular}{|l|l|l|l|l|l|l|l|l|l|l|l|l|l|l|l|}     \hline    \cellcolor{blue_color}$\star$ &\cellcolor{blue_color}$\star$ &\cellcolor{blue_color}$\star$ &\cellcolor{blue_color}$\star$ &\cellcolor{blue_color}$\star$ &\cellcolor{blue_color}$\star$ &\cellcolor{blue_color}$\star$ &\cellcolor{blue_color}$\star$ & 0 & 0 & 0 & 0 & 0 & 0 & 0 & 0 \\ \hline     \end{tabular}
    \normalsize\\
                                               &                                  &                     &                          \\
\multirow{2}{*}{$(0, \frac{n}{2}]$}            & ``is the key $> \frac{n}{2}$?''  & $1^{n/2} 0^{n/2}$   & 
    \small
    \tabcolsep=0.12cm
    \begin{tabular}{|l|l|l|l|l|l|l|l|l|l|l|l|l|l|l|l|}     \hline    \cellcolor{blue_color}$\star$ &\cellcolor{blue_color}$\star$ &\cellcolor{blue_color}$\star$ &\cellcolor{blue_color}$\star$ &\cellcolor{blue_color}$\star$ &\cellcolor{blue_color}$\star$ &\cellcolor{blue_color}$\star$ &\cellcolor{blue_color}$\star$ & 0 & 0 & 0 & 0 & 0 & 0 & 0 & 0 \\ \hline     \end{tabular}
    \normalsize\\
                                               & ``is the key $> \frac{n}{4}$?''  & $1^{n/4} 0^{3n/4}$  & 
                                                   \small
    \tabcolsep=0.12cm
                                               \begin{tabular}{|l|l|l|l|l|l|l|l|l|l|l|l|l|l|l|l|}     \hline    \cellcolor{blue_color}$\star$ &\cellcolor{blue_color}$\star$ &\cellcolor{blue_color}$\star$ &\cellcolor{blue_color}$\star$ &0 &0 &0 &0 & 0 & 0 & 0 & 0 & 0 & 0 & 0 & 0 \\ \hline     \end{tabular} 
                                               \normalsize\\
                                               &                                  &                     &                          \\
\multirow{2}{*}{$(\frac{n}{2}, n]$}            & ``is the key $> \frac{n}{2}$?''  & $1^{n/2} 0^{n/2}$   &     \small
    \tabcolsep=0.12cm
    \begin{tabular}{|l|l|l|l|l|l|l|l|l|l|l|l|l|l|l|l|}     \hline    \cellcolor{blue_color}$\star$ &\cellcolor{blue_color}$\star$ &\cellcolor{blue_color}$\star$ &\cellcolor{blue_color}$\star$ &\cellcolor{blue_color}$\star$ &\cellcolor{blue_color}$\star$ &\cellcolor{blue_color}$\star$ &\cellcolor{blue_color}$\star$ & 0 & 0 & 0 & 0 & 0 & 0 & 0 & 0 \\ \hline     \end{tabular}  \normalsize                 \\
                                               & ``is the key $> \frac{3n}{4}$?'' & $1^{3n/4} 0^{n/4}$  &     \small
    \tabcolsep=0.12cm
    \begin{tabular}{|l|l|l|l|l|l|l|l|l|l|l|l|l|l|l|l|}     \hline    \cellcolor{blue_color}$\star$ &\cellcolor{blue_color}$\star$ &\cellcolor{blue_color}$\star$ &\cellcolor{blue_color}$\star$ &\cellcolor{blue_color}$\star$ &\cellcolor{blue_color}$\star$ &\cellcolor{blue_color}$\star$ &\cellcolor{blue_color}$\star$ & \cellcolor{blue_color}$\star$ & \cellcolor{blue_color}$\star$ & \cellcolor{blue_color}$\star$ & \cellcolor{blue_color}$\star$ & 0 & 0 & 0 & 0 \\ \hline     \end{tabular}    \normalsize               \\
                                               &                                  &                     &                          \\
\multirow{2}{*}{$(0, \frac{n}{4}]$}            & ``is the key $> \frac{n}{4}$?''  & $1^{n/4} 0^{3n/4}$  &     \small
    \tabcolsep=0.12cm \begin{tabular}{|l|l|l|l|l|l|l|l|l|l|l|l|l|l|l|l|}     \hline    \cellcolor{blue_color}$\star$ &\cellcolor{blue_color}$\star$ &\cellcolor{blue_color}$\star$ &\cellcolor{blue_color}$\star$ &0 &0 &0 &0 & 0 & 0 & 0 & 0 & 0 & 0 & 0 & 0 \\ \hline     \end{tabular}  \normalsize                \\
                                               & ``is the key $> \frac{n}{8}$?''  & $1^{n/8} 0^{7n/8}$  & 
                                                   \small
    \tabcolsep=0.12cm
                                               \begin{tabular}{|l|l|l|l|l|l|l|l|l|l|l|l|l|l|l|l|}     \hline    \cellcolor{blue_color}$\star$ &\cellcolor{blue_color}$\star$ &0 &0 &0 &0 &0 &0 & 0 & 0 & 0 & 0 & 0 & 0 & 0 & 0 \\ \hline     \end{tabular} \normalsize                  \\
                                               &                                  &                     &                          \\
\multirow{3}{*}{$(\frac{n}{4}, \frac{n}{2}]$}  & ``is the key $> \frac{n}{4}$?''  & $1^{n/4} 0^{3n/4}$  &     \small
    \tabcolsep=0.12cm \begin{tabular}{|l|l|l|l|l|l|l|l|l|l|l|l|l|l|l|l|}     \hline    \cellcolor{blue_color}$\star$ &\cellcolor{blue_color}$\star$ &\cellcolor{blue_color}$\star$ &\cellcolor{blue_color}$\star$ &0 &0 &0 &0 & 0 & 0 & 0 & 0 & 0 & 0 & 0 & 0 \\ \hline     \end{tabular} \normalsize                   \\
                                               & ``is the key $> \frac{n}{2}$?''  & $1^{n/2} 0^{n/2}$                     & 
                                                   \small
    \tabcolsep=0.12cm
    \begin{tabular}{|l|l|l|l|l|l|l|l|l|l|l|l|l|l|l|l|}     \hline    \cellcolor{blue_color}$\star$ &\cellcolor{blue_color}$\star$ &\cellcolor{blue_color}$\star$ &\cellcolor{blue_color}$\star$ &\cellcolor{blue_color}$\star$ &\cellcolor{blue_color}$\star$ &\cellcolor{blue_color}$\star$ &\cellcolor{blue_color}$\star$ & 0 & 0 & 0 & 0 & 0 & 0 & 0 & 0 \\ \hline     \end{tabular} \normalsize                  \\
                                               & ``is the key $> \frac{3n}{8}$?'' & $1^{3n/8} 0^{5n/8}$ &     \small
    \tabcolsep=0.12cm \begin{tabular}{|l|l|l|l|l|l|l|l|l|l|l|l|l|l|l|l|}     \hline    \cellcolor{blue_color}$\star$ &\cellcolor{blue_color}$\star$ &\cellcolor{blue_color}$\star$ &\cellcolor{blue_color}$\star$ &\cellcolor{blue_color}$\star$ &\cellcolor{blue_color}$\star$ &0 &0 & 0 & 0 & 0 & 0 & 0 & 0 & 0 & 0 \\ \hline     \end{tabular}  \normalsize                 \\
                                               &                                  &                     &                          \\
\multirow{3}{*}{$(\frac{n}{2}, \frac{3n}{4}]$} & ``is the key $> \frac{n}{2}$?''  & $1^{n/2} 0^{n/2}$   &     \small
    \tabcolsep=0.12cm \begin{tabular}{|l|l|l|l|l|l|l|l|l|l|l|l|l|l|l|l|}     \hline    \cellcolor{blue_color}$\star$ &\cellcolor{blue_color}$\star$ &\cellcolor{blue_color}$\star$ &\cellcolor{blue_color}$\star$ &\cellcolor{blue_color}$\star$ &\cellcolor{blue_color}$\star$ &\cellcolor{blue_color}$\star$ &\cellcolor{blue_color}$\star$ & 0 & 0 & 0 & 0 & 0 & 0 & 0 & 0 \\ \hline     \end{tabular} \normalsize                   \\
                                               & ``is the key $> \frac{3n}{4}$?'' & $1^{3n/4} 0^{n/4}$  & 
                                                   \small
    \tabcolsep=0.12cm \begin{tabular}{|l|l|l|l|l|l|l|l|l|l|l|l|l|l|l|l|}     \hline    \cellcolor{blue_color}$\star$ &\cellcolor{blue_color}$\star$ &\cellcolor{blue_color}$\star$ &\cellcolor{blue_color}$\star$ &\cellcolor{blue_color}$\star$ &\cellcolor{blue_color}$\star$ &\cellcolor{blue_color}$\star$ &\cellcolor{blue_color}$\star$ & \cellcolor{blue_color}$\star$ & \cellcolor{blue_color}$\star$ & \cellcolor{blue_color}$\star$ & \cellcolor{blue_color}$\star$ & 0 & 0 & 0 & 0 \\ \hline     \end{tabular}  \normalsize                 \\
                                               & ``is the key $> \frac{5n}{8}$?'' & $1^{5n/8} 0^{3n/8}$ & 
                                                   \small
    \tabcolsep=0.12cm
    \begin{tabular}{|l|l|l|l|l|l|l|l|l|l|l|l|l|l|l|l|}     \hline    \cellcolor{blue_color}$\star$ &\cellcolor{blue_color}$\star$ &\cellcolor{blue_color}$\star$ &\cellcolor{blue_color}$\star$ &\cellcolor{blue_color}$\star$ &\cellcolor{blue_color}$\star$ &\cellcolor{blue_color}$\star$ &\cellcolor{blue_color}$\star$ & \cellcolor{blue_color}$\star$ & \cellcolor{blue_color}$\star$ & 0 & 0 & 0 & 0 & 0 & 0 \\ \hline     \end{tabular}     \normalsize              \\
                                               &                                  &                     &                          \\
\multirow{2}{*}{$(\frac{3n}{4}, n]$}           & ``is the key $> \frac{3n}{4}$?'' & $1^{3n/4} 0^{n/4}$  &     \small
    \tabcolsep=0.12cm \begin{tabular}{|l|l|l|l|l|l|l|l|l|l|l|l|l|l|l|l|}     \hline    \cellcolor{blue_color}$\star$ &\cellcolor{blue_color}$\star$ &\cellcolor{blue_color}$\star$ &\cellcolor{blue_color}$\star$ &\cellcolor{blue_color}$\star$ &\cellcolor{blue_color}$\star$ &\cellcolor{blue_color}$\star$ &\cellcolor{blue_color}$\star$ & \cellcolor{blue_color}$\star$ & \cellcolor{blue_color}$\star$ & \cellcolor{blue_color}$\star$ & \cellcolor{blue_color}$\star$ & 0 & 0 & 0 & 0 \\ \hline     \end{tabular} \normalsize                   \\
                                               & ``is the key $> \frac{7n}{8}$?'' & $1^{7n/8} 0^{n/8}$  &     \small
    \tabcolsep=0.12cm \begin{tabular}{|l|l|l|l|l|l|l|l|l|l|l|l|l|l|l|l|}     \hline    \cellcolor{blue_color}$\star$ &\cellcolor{blue_color}$\star$ &\cellcolor{blue_color}$\star$ &\cellcolor{blue_color}$\star$ &\cellcolor{blue_color}$\star$ &\cellcolor{blue_color}$\star$ &\cellcolor{blue_color}$\star$ &\cellcolor{blue_color}$\star$ & \cellcolor{blue_color}$\star$ & \cellcolor{blue_color}$\star$ & \cellcolor{blue_color}$\star$ & \cellcolor{blue_color}$\star$ & \cellcolor{blue_color}$\star$ & \cellcolor{blue_color}$\star$ & 0 & 0 \\ \hline     \end{tabular}    \normalsize              
\end{tabular}
\caption{The structure of random strings used in the communication protocol.}
\label{fig:os2-table}
\end{center}
\end{figure}
    
    The efficient communication protocol for the $\GT$ heavily relies on the algorithm above. In the communication protocol, both Alice and Bob run the algorithm to find the most significant bit where their inputs differ. Each time the algorithm asks ``is the position of the most significant bit where the inputs differ greater than $a$?'' (or ``is the key greater than $a$?''), Alice and Bob compute the equality of the first $a$ bits of their inputs to error $\frac{1}{4}$ by computing $\alpha = 2$ inner products of their inputs with random strings from $\R^{\tau}$ where $\tau = 1^a 0^{n-a}$. See Figure~\ref{fig:os2-table} for an illustration. If their inner products are the same then the answer to the algorithm is ``no, the key is not greater than $a$'' and otherwise, it's ``yes, the key is greater than $a$''. And at the end of the protocol, Alice and Bob check who has a greater value in the most significant bit discovered during the procedure.
    Now we are ready to prove the claim.

    \textbf{Constructing the \emph{good base} $\R_{\OS++}$.}
The construction is similar to the construction of $\R_{\OS}$.
    Let $\hat{\R}$ have the following structure.

\begin{align*}
\hat{\R} &=\R^{1^{n/2}0^{n/2}} \\
         &\times \left(\R^{1^{n/2}0^{n/2}} \times \R^{1^{n/4}0^{3n/4}}\right) \times \left(\R^{1^{n/2}0^{n/2}} \times \R^{1^{3n/4}0^{n/4}}\right)\\
         &\times \ldots \\
         &\times \left(\R^{1^{2}0^{n-2}} \times \R^{1^{1}0^{n-1}} \times \bigtimes_{a = 1}^{(n-4)/2}\left(\R^{1^{2a}0^{n-2a}} \times \R^{1^{2a+2}0^{n-2a-2}} \times \R^{1^{2a+1}0^{n-2a-1}} \right) \times \R^{1^{n-2}0^{2}} \times \R^{1^{n-1}0^{1}} \right)\\
         &\times \left(\bigtimes_{ 2c\log{n}}\left(\R^{1^{1}0^{n-1}} \times \bigtimes_{a = 1}^{n-2}\left(\R^{1^{a}0^{n-a}} \times \R^{1^{a+1}0^{n-a-1}}\right) \times \R^{1^{n-1}0^{1}} \right)\right)
\end{align*}

\begin{figure}[p]
\begin{center}
\begin{tabular}{cccc}
Interval                  & {\Large $\tau$}                    & Structure of $\R^{\tau}$ &                          \\
                          &                           &                          &                          \\
$(0, 8]$                  & $11110000$             & \begin{tabular}{|l|l|l|l|l|l|l|l|}     \hline    \cellcolor{blue_color}$\star$ &\cellcolor{blue_color}$\star$ &\cellcolor{blue_color}$\star$ &\cellcolor{blue_color}$\star$ & 0 & 0 & 0 & 0  \\ \hline     \end{tabular}                   &                          \\
                          &                           &                          &                          \\
\multirow{2}{*}{$(0, 4]$} & $1111 0000$             & \begin{tabular}{|l|l|l|l|l|l|l|l|}     \hline    \cellcolor{blue_color}$\star$ &\cellcolor{blue_color}$\star$ &\cellcolor{blue_color}$\star$ &\cellcolor{blue_color}$\star$ & 0 & 0 & 0 & 0  \\ \hline     \end{tabular}                   &                          \\
                          & $11000000$             & \begin{tabular}{|l|l|l|l|l|l|l|l|}     \hline    \cellcolor{blue_color}$\star$ &\cellcolor{blue_color}$\star$ & 0 & 0 & 0 & 0 & 0 & 0  \\ \hline     \end{tabular}                   &                          \\
                          &                           &                          &                          \\
\multirow{2}{*}{$(4, 8]$} & $11110000$             & \begin{tabular}{|l|l|l|l|l|l|l|l|}     \hline    \cellcolor{blue_color}$\star$ &\cellcolor{blue_color}$\star$ &\cellcolor{blue_color}$\star$ &\cellcolor{blue_color}$\star$ & 0 & 0 & 0 & 0  \\ \hline     \end{tabular}                   &                          \\
                          & $11111100$             & \begin{tabular}{|l|l|l|l|l|l|l|l|}     \hline    \cellcolor{blue_color}$\star$ &\cellcolor{blue_color}$\star$ &\cellcolor{blue_color}$\star$ &\cellcolor{blue_color}$\star$ & \cellcolor{blue_color}$\star$ & \cellcolor{blue_color}$\star$ & 0 & 0  \\ \hline     \end{tabular}                   &                          \\
                          &                           &                          &                          \\
\multirow{2}{*}{$(0, 2]$} & $11000000$             & \begin{tabular}{|l|l|l|l|l|l|l|l|}     \hline    \cellcolor{blue_color}$\star$ &\cellcolor{blue_color}$\star$ & 0 & 0 & 0 & 0 & 0 & 0  \\ \hline     \end{tabular}                   &                          \\
                          & $10000000$             & \begin{tabular}{|l|l|l|l|l|l|l|l|}     \hline    \cellcolor{blue_color}$\star$ & 0 & 0 & 0 & 0 & 0 & 0 & 0  \\ \hline     \end{tabular}                   &                          \\
                          &                           &                          &                          \\
\multirow{3}{*}{$(2, 4]$} & $11000000$             & \begin{tabular}{|l|l|l|l|l|l|l|l|}     \hline    \cellcolor{blue_color}$\star$ &\cellcolor{blue_color}$\star$ & 0 & 0 & 0 & 0 & 0 & 0  \\ \hline     \end{tabular}                   &                          \\
                          &          $11110000$                 & \begin{tabular}{|l|l|l|l|l|l|l|l|}     \hline    \cellcolor{blue_color}$\star$ &\cellcolor{blue_color}$\star$ &\cellcolor{blue_color}$\star$ &\cellcolor{blue_color}$\star$ & 0 & 0 & 0 & 0  \\ \hline     \end{tabular}                   &                          \\
                          & $11100000$             & \begin{tabular}{|l|l|l|l|l|l|l|l|}     \hline    \cellcolor{blue_color}$\star$ &\cellcolor{blue_color}$\star$ &\cellcolor{blue_color}$\star$ & 0 & 0 & 0 & 0 & 0  \\ \hline     \end{tabular}                   &                          \\
                          &                           &                          &                          \\
\multirow{3}{*}{$(4, 6]$} & $11110000$             & \begin{tabular}{|l|l|l|l|l|l|l|l|}     \hline    \cellcolor{blue_color}$\star$ &\cellcolor{blue_color}$\star$ &\cellcolor{blue_color}$\star$ &\cellcolor{blue_color}$\star$ & 0 & 0 & 0 & 0  \\ \hline     \end{tabular}                   &                          \\
                          & $11111100$             & \begin{tabular}{|l|l|l|l|l|l|l|l|}     \hline    \cellcolor{blue_color}$\star$ &\cellcolor{blue_color}$\star$ &\cellcolor{blue_color}$\star$ &\cellcolor{blue_color}$\star$ & \cellcolor{blue_color}$\star$ & \cellcolor{blue_color}$\star$ & 0 & 0  \\ \hline     \end{tabular}                   &                          \\
                          & $11111000$             & \begin{tabular}{|l|l|l|l|l|l|l|l|}     \hline    \cellcolor{blue_color}$\star$ &\cellcolor{blue_color}$\star$ &\cellcolor{blue_color}$\star$ &\cellcolor{blue_color}$\star$ & \cellcolor{blue_color}$\star$ & 0 & 0 & 0  \\ \hline     \end{tabular}                   &                          \\
                          &                           &                          &                          \\
\multirow{2}{*}{$(6, 8]$} & $11111100$             & \begin{tabular}{|l|l|l|l|l|l|l|l|}     \hline    \cellcolor{blue_color}$\star$ &\cellcolor{blue_color}$\star$ &\cellcolor{blue_color}$\star$ &\cellcolor{blue_color}$\star$ & \cellcolor{blue_color}$\star$ & \cellcolor{blue_color}$\star$ & 0 & 0  \\ \hline     \end{tabular}                   &                          \\
                          & $11111110$             & \begin{tabular}{|l|l|l|l|l|l|l|l|}     \hline    \cellcolor{blue_color}$\star$ &\cellcolor{blue_color}$\star$ &\cellcolor{blue_color}$\star$ &\cellcolor{blue_color}$\star$ & \cellcolor{blue_color}$\star$ & \cellcolor{blue_color}$\star$ & \cellcolor{blue_color}$\star$ & 0  \\ \hline     \end{tabular}                   &                          \\
                          &                           &                          &                          \\
\multirow{3}{*}{$(0, 1]$} & \multirow{3}{*}{$10000000 \Biggl\{$}              & \begin{tabular}{|l|l|l|l|l|l|l|l|}     \hline    \cellcolor{blue_color}$\star$ & 0 & 0 & 0 & 0 & 0 & 0 & 0  \\ \hline     \end{tabular}                   & \multirow{3}{*}{$\Biggr\}$ repeat $2c \log{n}$ times} \\
                          &  & $\ldots$                 &                          \\
                          &                           & \begin{tabular}{|l|l|l|l|l|l|l|l|}     \hline    \cellcolor{blue_color}$\star$ &0 &0 &0 & 0 & 0 & 0 & 0  \\ \hline     \end{tabular}                   &                          \\
                          &                           &                          &                          \\
\multirow{6}{*}{$(1, 2]$} & \multirow{3}{*}{$10000000 \Biggl\{$}              & \begin{tabular}{|l|l|l|l|l|l|l|l|}     \hline    \cellcolor{blue_color}$\star$ & 0 & 0 & 0 & 0 & 0 & 0 & 0  \\ \hline     \end{tabular}                   & \multirow{3}{*}{$\Biggr\}$ repeat $2c \log{n}$ times} \\
                          &  & $\ldots$                 &                          \\
                          &                           & \begin{tabular}{|l|l|l|l|l|l|l|l|}     \hline    \cellcolor{blue_color}$\star$ & 0 & 0 & 0 & 0 & 0 & 0 & 0  \\ \hline     \end{tabular}                   &                          \\ \\
                          & \multirow{3}{*}{$11000000 \Biggl\{$}              & \begin{tabular}{|l|l|l|l|l|l|l|l|}     \hline    \cellcolor{blue_color}$\star$ &\cellcolor{blue_color}$\star$ &0 &0 & 0 & 0 & 0 & 0  \\ \hline     \end{tabular}                   & \multirow{3}{*}{$\Biggr\}$ repeat $2c \log{n}$ times} \\
                          &  & $\ldots$                 &                          \\
                          &                           & \begin{tabular}{|l|l|l|l|l|l|l|l|}     \hline    \cellcolor{blue_color}$\star$ &\cellcolor{blue_color}$\star$ &0 &0 & 0 & 0 & 0 & 0  \\ \hline     \end{tabular}                   &                          \\
\multicolumn{3}{c}{{\Huge$\ldots$}}                                                     &                          \\
\multirow{3}{*}{$(7, 8]$} & \multirow{3}{*}{$11111110 \Biggl\{$}              & \begin{tabular}{|l|l|l|l|l|l|l|l|}     \hline    \cellcolor{blue_color}$\star$ &\cellcolor{blue_color}$\star$ &\cellcolor{blue_color}$\star$ &\cellcolor{blue_color}$\star$ & \cellcolor{blue_color}$\star$ & \cellcolor{blue_color}$\star$ & \cellcolor{blue_color}$\star$ & 0  \\ \hline     \end{tabular}                   & \multirow{3}{*}{$\Biggr\}$ repeat $2c \log{n}$ times} \\
                          & & $\ldots$                 &                          \\
                          &                           & \begin{tabular}{|l|l|l|l|l|l|l|l|}     \hline    \cellcolor{blue_color}$\star$ &\cellcolor{blue_color}$\star$ &\cellcolor{blue_color}$\star$ &\cellcolor{blue_color}$\star$ & \cellcolor{blue_color}$\star$ & \cellcolor{blue_color}$\star$ & \cellcolor{blue_color}$\star$ & 0  \\ \hline     \end{tabular}                   &                         
\end{tabular}
\end{center}
\caption{Structure of $\hat{\R}$. Blue cells with $\star$ represent indices where either a 0 or 1 could appear.}
\label{fig:struct-os2}
\end{figure}
See Figure~\ref{fig:struct-os2} for an illustration. This $\hat{\R}$ describes all the strings as the source of randomness needed for the $O(\log{n})$ communication protocol for $\GT$, but each of the strings appears in the structure only once instead of $\alpha c \log{n}$ times. So, we need to duplicate this structure $\alpha c \log{n}$ times to properly simulate the protocol.

As in the proof of Claim~\ref{clm:os-ub}, we are going to add two other steps to the structure. 
First, we are going to have $t$ individual ``prepackaged'' copies for the $\GT$ protocol. Let $\R_1 = \ldots = \R_{t} = \times_{\alpha c  \log{n}} \hat{\R}$. Each of the copies has enough randomness and the right structure of that randomness to simulate one full run of the $\GT$ protocol. Let $\R' = \bigtimes_{j \in [t]} \R_{j}$, which allows us to handle $t$ runs.
Secondly, we add a set of ``basis'' strings to the structure: $\mathcal{B} = \{\one_1\} \times \{\one_2\} \times \ldots \{\one_{n-1}\} \times \{\one_n\} = \{10\ldots0\}\times\{010\ldots0\} \times \ldots \times \{00\ldots010\} \times \{00\ldots01\}$. 

The final underlying structure of the oracle will be a Cartesian product of $\R'$ and $\mathcal{B}$: 
$\R_{\OS++} = \mathcal{B} \times \R' = \mathcal{B} \times (\bigtimes_{j \in [t]} \R_{j})$. See Figure~\ref{fig:oracle-os2} for an illustration.

\begin{figure}[h]
\begin{center}
\begin{tabular}{ccccccccc}
$\mathcal{B}$ &             & $\R_1$ &             & $\R_2$ &             &          &             & $\R_t$ \\
 &             & $c \alpha \log{n}$ copies &             & $c \alpha \log{n}$ copies &             &          &             & $c \alpha \log{n}$ copies \\
\footnotesize
\tabcolsep=0.11cm
\begin{tabular}{|llllll|}
\hline
\cellcolor{yellow}1 & 0 & 0 & 0 & 0 & 0 \\
0 & \cellcolor{yellow}1 & 0 & 0 & 0 & 0 \\
0 & 0 & \cellcolor{yellow}1 & 0 & 0 & 0 \\
0 & 0 & 0 & \cellcolor{yellow}1 & 0 & 0 \\
0 & 0 & 0 & 0 & \cellcolor{yellow}1 & 0 \\
0 & 0 & 0 & 0 & 0 & \cellcolor{yellow}1 \\ \hline
\end{tabular}
& $\bigtimes$ &
\footnotesize
\tabcolsep=0.11cm
\begin{tabular}{|llllllll|}
\hline
$\cellcolor{blue_color}$ & $\cellcolor{blue_color}$ & $\cellcolor{blue_color}$ & $\cellcolor{blue_color}$ & & & & \\

$\cellcolor{blue_color}$ & $\cellcolor{blue_color}$ & $\cellcolor{blue_color}$ & $\cellcolor{blue_color}$ & & & & \\
$\cellcolor{blue_color}$ & $\cellcolor{blue_color}$ & & & & & & \\

$\cellcolor{blue_color}$ & $\cellcolor{blue_color}$ & $\cellcolor{blue_color}$ & $\cellcolor{blue_color}$ & & & & \\
$\cellcolor{blue_color}$ & $\cellcolor{blue_color}$ & $\cellcolor{blue_color}$ & $\cellcolor{blue_color}$ & $\cellcolor{blue_color}$ & $\cellcolor{blue_color}$ & & \\

$\cellcolor{blue_color}$ & $\cellcolor{blue_color}$ & & & & & & \\
$\cellcolor{blue_color}$ & & & & & & & \\
$\cellcolor{blue_color}$ & $\cellcolor{blue_color}$ & & & & & & \\
$\cellcolor{blue_color}$ & $\cellcolor{blue_color}$ & $\cellcolor{blue_color}$ & $\cellcolor{blue_color}$ & & & & \\
$\cellcolor{blue_color}$ & $\cellcolor{blue_color}$ & $\cellcolor{blue_color}$ &  &  &  &  & \\
&\multicolumn{6}{c}{{\LARGE$\ldots$}}&\\
\hline
\end{tabular}
& $\bigtimes$ & 
\footnotesize
\tabcolsep=0.11cm
\begin{tabular}{|llllllll|}
\hline
$\cellcolor{blue_color}$ & $\cellcolor{blue_color}$ & $\cellcolor{blue_color}$ & $\cellcolor{blue_color}$ & & & & \\

$\cellcolor{blue_color}$ & $\cellcolor{blue_color}$ & $\cellcolor{blue_color}$ & $\cellcolor{blue_color}$ & & & & \\
$\cellcolor{blue_color}$ & $\cellcolor{blue_color}$ & & & & & & \\

$\cellcolor{blue_color}$ & $\cellcolor{blue_color}$ & $\cellcolor{blue_color}$ & $\cellcolor{blue_color}$ & & & & \\
$\cellcolor{blue_color}$ & $\cellcolor{blue_color}$ & $\cellcolor{blue_color}$ & $\cellcolor{blue_color}$ & $\cellcolor{blue_color}$ & $\cellcolor{blue_color}$ & & \\

$\cellcolor{blue_color}$ & $\cellcolor{blue_color}$ & & & & & & \\
$\cellcolor{blue_color}$ & & & & & & & \\
$\cellcolor{blue_color}$ & $\cellcolor{blue_color}$ & & & & & & \\
$\cellcolor{blue_color}$ & $\cellcolor{blue_color}$ & $\cellcolor{blue_color}$ & $\cellcolor{blue_color}$ & & & & \\
$\cellcolor{blue_color}$ & $\cellcolor{blue_color}$ & $\cellcolor{blue_color}$ &  &  &  &  & \\
&\multicolumn{6}{c}{{\LARGE$\ldots$}}&\\
\hline
\end{tabular}
& $\bigtimes$ & $\ldots$ & $\bigtimes$ & 
\footnotesize
\tabcolsep=0.11cm
\begin{tabular}{|llllllll|}
\hline
$\cellcolor{blue_color}$ & $\cellcolor{blue_color}$ & $\cellcolor{blue_color}$ & $\cellcolor{blue_color}$ & & & & \\

$\cellcolor{blue_color}$ & $\cellcolor{blue_color}$ & $\cellcolor{blue_color}$ & $\cellcolor{blue_color}$ & & & & \\
$\cellcolor{blue_color}$ & $\cellcolor{blue_color}$ & & & & & & \\

$\cellcolor{blue_color}$ & $\cellcolor{blue_color}$ & $\cellcolor{blue_color}$ & $\cellcolor{blue_color}$ & & & & \\
$\cellcolor{blue_color}$ & $\cellcolor{blue_color}$ & $\cellcolor{blue_color}$ & $\cellcolor{blue_color}$ & $\cellcolor{blue_color}$ & $\cellcolor{blue_color}$ & & \\

$\cellcolor{blue_color}$ & $\cellcolor{blue_color}$ & & & & & & \\
$\cellcolor{blue_color}$ & & & & & & & \\
$\cellcolor{blue_color}$ & $\cellcolor{blue_color}$ & & & & & & \\
$\cellcolor{blue_color}$ & $\cellcolor{blue_color}$ & $\cellcolor{blue_color}$ & $\cellcolor{blue_color}$ & & & & \\
$\cellcolor{blue_color}$ & $\cellcolor{blue_color}$ & $\cellcolor{blue_color}$ &  &  &  &  & \\
&\multicolumn{6}{c}{{\LARGE$\ldots$}}&\\
\hline
\end{tabular}
\end{tabular}
\caption{Structure of $\R_{\OS++}$. Each $\R_j$ consists of $c \alpha \log{n}$ copies of $\hat{\R}$.}
\label{fig:oracle-os2}
\end{center}
\end{figure}

\normalsize

We also set the parameters to be $\alpha = 2, t = 250n\ln{2}$.
Notice that this set $\R_{\OS++}$ is a \emph{good base} by construction and has size $m \leq n + 3 \alpha t (c n \log{n} + 2 c n \log^2{n}) = O(n^3 \log^2{n})$.
    
    \textbf{Constructing the family of approximating polynomials.}
    The construction will follow the same outline as the construction of polynomials for $\R_{\OS++}$. We construct a family of deterministic algorithms that work well for every fixed pair of inputs $i, x\in \{0, 1\}^n$, and then construct a family of algorithms that work for all inputs with high probability at the same time, and then finally explain how to convert it to a family of approximating polynomials.
            
    For all $i \in \{0, 1\}^n, j \in [t], r \in \R_{\OS++}$ let $B_{(r, i, j)}(\Oracle(r, x))$ be the following deterministic algorithm.

    \begin{enumerate}
        \item Simulate the algorithm from Claim~\ref{clm:comp-tree}.
        \item Each time the algorithm asks ``is the key $> a$?'':
        \item \qquad Set $\tau = 1^{a} 0^{n - a}$.
        \item \qquad For the $\alpha$ indices $v \in [m]$ corresponding to $n$-bit strings drawn from $\R^{\tau}$ within the $j$th 
        
        \qquad copy $\R_j$ that were not already used prior to this step:
        \item \qquad \qquad Compute $\langle i, r_v \rangle$ and compare it to $(\Oracle(r, x))_v = \langle x, r_v \rangle$. 
        \item \qquad If for all $\alpha$ indices $v$ the inner products are the same then reply ``yes, the key $> a$''. 
        
        \qquad Otherwise, reply ``no, the key $\leq a$''.
        \item Let $k$ be the output of the algorithm.
        \item Compare $i_k = \langle i, \one_k \rangle$ and $(\Oracle(r, x))_k = \langle x, \one_k \rangle = x_k$. If $x_k \leq i_k$ then accept. Otherwise, reject.
    \end{enumerate}
    Notice that this algorithm emulates the $O(\log{n})$ randomized communication protocol for $\GT$ communication problem. 
   
    This algorithm computes $\GT_i(x)$ with probability at least $\frac{11}{12}$ for all $j \in [t]$. That is, for all $j \in [t]$ and for all $i, x \in \{0, 1\}^n$

    $$\Pr_{r \from \R_{\OS++}}[B_{(r, i, j)}(\Oracle(r, x)) = \GT_i(x)] \geq \frac{11}{12}.$$
    This algorithm makes at most $3 \alpha c \log{n} = 6 c \log{n} $ queries to the oracle $\Oracle(r, x)$. 
    Let $W(i, x, r, j)$ be the indicator that the $j$-th package of random strings in $r$ defines a set of ``bad'' random strings for $(i, x)$: $W(i, x, r, j) = 1$ if and only if $B_{(r, i, j)}(\Oracle(r, x)) \neq \GT_i(x)$. We established that $B_{(r, i, j)}(\Oracle(r, x))$ works well if given a random $r \from \R_{\OS++}$ for every $j \in [t]$ and the probability of this algorithm outputting an incorrect answer is at most $\frac{1}{12}$. So for all $i, x \in \{0, 1\}^n, j \in [t]$ 
    $$\Pr_{r \from \R_{\OS++}}[W(i, x, r, j) = 1] = \E_{r \from \R_{\OS++}}[W(i, x, r, j)] \leq \frac{1}{12}.$$

    For all $i \in \{0, 1\}^n, r \in \R_{\OS++}$ let $A_{(r, i)}(\Oracle(r, x))$ be the following randomized algorithm:
        
    \begin{itemize}
        \item Choose $j \from [t]$ uniformly at random.
        \item Run $B_{(r, i, j)}(\Oracle(r, x))$.
    \end{itemize}
    The number of queries that $A_{(r, i)}$ makes to the oracle is the same as $B_{(r, i, j)}$ which is $3\alpha c \log{n} = 6c\log{n}$. We fix a pair $(i, x)$ and evaluate the following probability.
    
    $$ \Pr_{r \from \R_{\OS++}}\left[\Pr_{j \from [t]}[B_{(r, i, j)} \neq \GT_i(x)] > \frac{1}{6}\right] = \Pr_{r \from \R_{\OS++}}\left[\frac{1}{t} \sum_{j \in [t]} W(i, x, r, j) > \frac{1}{6}\right].$$
    We established that $\E_{r \from \R_{\OS++}}[W(i, x, r, j)] \leq \frac{1}{12}$ and so by Hoeffding's inequality,
    $$\Pr_{r \from \R_{\OS++}}\left[\frac{1}{t} \sum_{j \in [t]} W(i, x, r, j) > \frac{1}{12} + \frac{1}{12}\right] \leq e^{-2\frac{t}{144}} \leq 2^{-\frac{500n}{144}}.$$
    By a union bound over all possible $i, x \in \{0, 1\}^n$,
    $$\Pr_{r \from \R_{\OS++}}\left[\exists i, x \in \{0, 1\}^n: \frac{1}{t} \sum_{j \in [t]} W(i, x, r, j) > \frac{1}{6}\right] \leq 2^{2n} 2^{-\frac{500n}{144}} \leq 2^{-n} < \frac{1}{3}.$$
    Therefore, we have proven that
    $$\Pr_{r \from \R_{\OS++}}\left[\exists i, x \in \{0, 1\}^n: \Pr_{j \from [t]}[A_{(r, i)}(\Oracle(r, x))) \neq \GT_i(x)] > \frac{1}{6}\right] < \frac{1}{3}.$$
        
    The last step is to convert this family of query algorithms into a family of approximating polynomials. Let $q_{(r, i)}$ denote the acceptance probability of $A_{(r, i)}$ which is a polynomial of degree at most $6c \log{n}$ such that
        
    $$\Pr_{r \from \R_{\OS++}}\left[\exists i, x \in \{0, 1\}^n: \left|q_{(r, i)}(\Oracle(r, x)) - \GT_i(x)\right| > \frac{1}{6}\right] < \frac{1}{3}$$
    which is exactly what we were looking for.

\end{appendix}

\end{document}